\documentclass[11pt, letterpaper]{amsart}
\usepackage{amsfonts}
\usepackage{amsmath}
\usepackage{color}
\usepackage{enumerate}
\usepackage{amssymb}
\usepackage{graphicx}
\usepackage{appendix}
\usepackage{ulem}

\newtheorem{thm}{Theorem}[section]
\newtheorem{cor}[thm]{Corollary}
\newtheorem{lem}[thm]{Lemma}
\newtheorem{prop}[thm]{Proposition}
\theoremstyle{definition}
\newtheorem{defn}[thm]{Definition}

\newtheorem{ass}[thm]{Assumption}
\newtheorem{tass}[thm]{Temporary Assumption}

\theoremstyle{remark}
\newtheorem{rem}[thm]{Remark}

\numberwithin{equation}{section}

\newcommand{\filt}{\mathbb{F}}
\newcommand{\filtg}{\mathbb{G}}

\newcommand{\prob}{\mathbb{P}}
\newcommand{\tprob}{\widetilde{\mathbb{P}}}
\newcommand{\qprob}{\mathbb{Q}}
\newcommand{\reals}{\mathbb R}

\newcommand{\A}{\mathcal{A}}
\newcommand{\B}{\mathcal{B}}

\newcommand{\E}{\mathcal{E}}
\newcommand{\F}{\mathcal{F}}
\newcommand{\G}{\mathcal{G}}
\newcommand{\Hcal}{\mathcal{H}}

\newcommand{\OO}{\mathcal{O}}
\newcommand{\mcp}{\mathcal{P}}

\newcommand{\We}{\mathcal{W}}
\newcommand{\Y}{\mathcal{Y}}

\newcommand{\cz}{\check{Z}}
\newcommand{\eps}{\varepsilon}
\newcommand{\such}{\ | \ }

\newcommand{\probtriple}{(\Omega, \mathcal{F}, \mathbb{P})}
\newcommand{\basis}{(\Omega, \mathcal{F}, \mathbb{F}, \mathbb{P})}

\newcommand{\dfn}{\, := \,}
\newcommand{\rdfn}{\, =: \,}

\newcommand{\indep}{\perp\!\!\!\perp}

\newcommand{\util}{\mathcal{U}}
\newcommand{\imutil}{\mathcal{I}}
\newcommand{\dualutil}{\mathcal{V}}

\newcommand{\expv}[3]{\mathbb{E}^{#1}_{#2}\left[#3\right]}

\newcommand{\condexpv}[4]{\mathbb{E}^{#1}_{#2}\left[#3\big| #4\right]}

\newcommand{\expvs}[1]{\mathbb{E}\left[#1\right]}
\newcommand{\condexpvs}[2]{\mathbb{E}\left[#1\big| #2\right]}

\newcommand{\condprobs}[2]{\mathbb{P}\left[#1\big| #2\right]}

\newcommand{\nada}[1]{}

\newcommand{\wt}[1]{\widetilde{#1}}
\newcommand{\wh}[1]{\widehat{#1}}

\newcommand{\bra}[1]{\left[#1\right]}
\newcommand{\cbra}[1]{\left\{#1\right\}}

\newcommand{\ul}[1]{\underline{#1}}

\newcommand{\prodlog}[1]{\mathrm{PL}\left(#1\right)}

\def\jd{\textcolor{blue}}

\begin{document}
\title[vNM $U$]{Dynamic Equilibrium with Insider Information and General Uninformed Agent Utility}

\author{Jerome Detemple}
\author{Scott Robertson}
\address{Questrom School of Business, Boston, University, Boston MA, 02215, USA}
\email{detemple@bu.edu; scottrob@bu.edu}

\date{\today }

\begin{abstract}
    We study a continuous time economy where agents have asymmetric information.  The informed agent (``$I$''), at time zero, receives a private signal about the risky assets' terminal payoff $\Psi(X_T)$, while the uninformed agent (``$U$'') has no private signal.  $\Psi$ is an arbitrary payoff function, and $X$ follows a time-homogeneous diffusion. Crucially, we allow $U$ to have von Neumann-Morgenstern preferences with a general utility function on $(0,\infty)$  satisfying the standard conditions. This extends previous constructions of equilibria with asymmetric information used when all agents have exponential utilities and  enables us to study the impact of $U$'s initial share endowment on equilibrium. To allow for $U$ to have general preferences, we introduce a new method to prove existence of a partial communication equilibrium (PCE), where at time $0$, $U$ receives a less-informative signal than $I$. In the single asset case, this signal is  recoverable by viewing the equilibrium price process over an arbitrarily short period of time, and hence the PCE is a dynamic noisy rational expectations equilibrium. Lastly, when $U$ has power (constant relative risk aversion) utility, we identify the equilibrium price in the small and large risk aversion limits.
\end{abstract}

\maketitle

\section{Introduction}
The study of economies with asymmetric information has historically focused on the static Exponential-Gaussian framework with exogenous noise trading. In this model uninformed agents ($U$) and informed agents ($I$) have exponential (constant absolute risk aversion or ``CARA'') utility, and the private signal, asset terminal payoff and noise trader demand are jointly normally distributed. Recent advances extend results beyond Gaussian terminal payoffs (\cite{breon2015existence}) and to dynamic, diffusion-based models where noise traders display  endogenous behavior (\cite{MR4192554}). However, the CARA utility structure is retained. Presently, we generalize \cite{MR4192554} to allow the uninformed agent to have von Neumann-Morgenstern (vNM) preferences with a general utility function on $(0,\infty)$ satisfying standard conditions. In this setting, we examine the existence and properties of equilibria.

The information conveyed by prices plays a fundamental role in economics. It determines the investment opportunities of agents in markets and informs managerial decision-making within firms. The degree of informational efficiency of prices, i.e., the extent to which prices reveal information disseminated in the economy, is widely thought to be a critical ingredient for smooth market and economic operations. Yet, studies of economies with asymmetric information, even those providing recent advances, remain limited to the CARA utility model, thus raising questions about the broad applicability of conclusions. Generalizations along the preference dimension are therefore desirable in order to assess the robustness and limits of prior results. This paper takes a step in this direction. While it maintains the CARA behavior of informed agents $I$ and noise traders $N$, it allows for uninformed agents $U$ with vNM preferences where the utility function satisfies standard monotonicity, concavity, Inada and asymptotic elasticity conditions.

The paper contributes on three fronts. First, we prove the existence of a partial communication equilibrium (PCE) in which a noisy version of the informed agent's signal is released at the outset. The PCE (as opposed to an equilibrium where all agents have identical information) is of interest as it preserves the informed agent's informational advantage. This result is valid under mild conditions upon the underlying diffusion $X$, terminal payoff function $\Psi$, uninformed agent utility function, and initial share endowments, see Assumptions \ref{A:X_SDE}, \ref{A:Psi}, \ref{A:U}, and \ref{A:init_pos} below. An important feature of this equilibrium, and one that does not hold under  CARA utility, is that when the uninformed agent has a general utility function, his policy depends on his initial wealth, and hence on the initial price vector, as agents are endowed with exogenous initial share endowments, rather than exogenous initial wealth. It follows that the equilibrium price, as well as its volatility and the market price of risk, at any date, depend on the initial endowment as well. This feedback effect leads to a non-trivial fixed point problem for the stochastic discount factor, that we resolve. In addition, the lack of an explicit optimal portfolio formula for the uninformed renders ineffective previous approaches for the construction of equilibria with asymmetric information. We introduce a new method to overcome this difficulty. Therefore, our results enable the study of endogenous endowment effects on equilibrium quantities, something that is not possible in the CARA setting.  Identifying endowment effects is our primary motivation for extending beyond CARA preferences, and as such, our results are not simply a technical extension of the CARA setting.

Second, focusing on the single asset case, we show the PCE is in fact a dynamic noisy rational expectations equilibrium (DNREE). What this means is that the filtration used by $U$ (the market filtration) coincides with the right filtration generated by the underlying economic factor process (diffusion $X$) and the equilibrium price process. This in turn, provides a mechanism for information communication through the endogenously determined equilibrium price. That the endogenous price instantaneously reveals the public signal holds because the price path over any arbitrarily small interval $[0,\eps]$ cannot be the same for different realizations of the signal.



Third, we derive a quasi-explicit (in terms of the Lambert function) solution when $U$ has constant relative risk aversion (CRRA), i.e., when the economy is comprised of a mixture of CRRA and CARA agents. In this context, we identify the equilibrium price in the limit as relative risk aversion either explodes or vanishes. When $U$'s risk aversion goes to infinity, the agent seeks certainty, implying equilibrium is driven by the behavior of the two CARA agents. The limit equilibrium price is then the same as in an economy with just the two CARA agents. When $U$'s risk aversion vanishes, the limit equilibrium price does not correspond to the risk neutral price, as might have been intuitively expected, and is the case when all agents have CARA utilities. The reason it departs from it is because CRRA utility enforces a non-negative wealth constraint, which does not vanish in the small risk aversion limit.

Methodologically, we contribute by relating an equilibrium where agents have differing information sets, to one where agents have identical information but differ in their beliefs and/or random endowments.  Heuristically, this is because for two random variables $X,G$ such that $\condprobs{G\in dg}{X} \sim \prob\bra{G \in dg}$ almost surely with density $p^g(X)$,  for a given function $\phi$ one has $\condexpvs{\phi(X,G)}{G}  = f(G;\phi)$ where $f(g;\phi) = \expvs{\phi(X,g)p^g(X)}$. Thus, if one views $\phi$ as a control lying in a set $\A(G)$ of controls which may depend upon $G$, it is natural to solve the problem $\sup_{\phi\in \A(G)}\condexpvs{\phi(X,G)}{G} = \sup_{\phi\in \A(G)} f(G;\phi)$ by first solving the problem $\sup_{\phi\in \A(g)} f(g;\phi)$ for fixed $g$, even if $\prob\bra{G=g}= 0$ for all $g$. This perspective motivates us to first establish the existence of a ``signal-realization'' equilibrium where agents have common information, using classical results on optimal investments in complete markets.\footnote{See Remark \ref{R:signal_realization} for additional discussion on how we implement the above idea.}  We then lift the signal realization equilibrium to the signal level, by using delicate measurability arguments for parameter driven stochastic integrals, as can be found in \cite{MR510530,MR1775229, MR3758346,MR4192554}. The key feature of our method is that it does not impose global (i.e., at the signal level) integrability conditions upon the insider and uninformed agents terminal wealth. This is crucial, because unlike when all agents have CARA preferences, the uninformed agent's optimal position is no longer affine in the market signal and in fact we do not have explicit knowledge of the ``signal to strategy'' map. As a result, it becomes prohibitively difficult to verify the standard global integrability conditions used in optimal investment problems. However, conceptually it does not make sense to require global integrability, as the agents are conditioning upon the signals and hence the signals should be thought of as constants. Our method of lifting from the signal-realization level to the signal level is used precisely to incorporate this thinking and overcome the difficulty. We hope this methodology can be used in the future to extend beyond the CARA insider case.

The paper relates to three branches of the asset pricing theory literature. First, it connects directly with \cite{MR4192554}, which proves existence and properties of equilibria in dynamic CARA-diffusion models with bounded rational noise trading. That paper shows in particular existence of a DNREE where the initial price (vector) reveals a noisy aggregate of the informed (vector) signal if the (vector) payoff function is linear in the underlying (vector) factor, or in the one asset case if it is monotonic with respect to the single underlying factor. Here, we extend the analysis by allowing for uninformed agents with general utility functions satisfying standard Inada and elasticity conditions and arbitrary payoff function satisfying the Jacobian condition. As indicated above, this extension implies feedback effects from state prices to endowment values and introduces non-trivial integrability issues, leading to substantial complications for showing the existence of equilibria with endogenous information structures. We also show that the PCE is a DNREE in the weaker sense of equivalence of right filtrations, in the single asset case, and prove that a vehicle for instantaneous revelation is the asset price. Equilibria display the same informational content as in the pure CARA model, establishing robustness along the uninformed agent's preferences dimension. An additional aspect of our contribution is to permit the analysis of behavioral traits of uninformed agents, such as risk aversion, on equilibrium prices and allocations. Our results uncover differences in behavior relative to the CARA model.

Second, it relates more generally to the broad literature on asymmetric information. Pioneering contributions by \cite{grossman1976efficiency} and \cite{hellwig1980aggregation} respectively show the existence of a fully revealing and a noisy rational expectations equilibrium in the static CARA-Gaussian setting. An early departure from that classic setting can be found in \cite{radner1979rational}, a paper that introduces the notion of full communication equilibrium (FCE) and establishes its relation to the concept of rational expectations equilibrium. Our notion of PCE, introduced in \cite{MR4192554}, is a direct extension of FCE. Recent contributions, while maintaining the CARA setting, provide extensions to non-Gaussian distributions. Notable studies along these lines include \cite{breon2015existence} and \cite{chabakauri2021}. Our analysis differs in terms of generality pertaining to the uncertainty structure, the uninformed agent's preferences with non-CARA utility, and the temporal setting. The second aspect implies price dynamics that depend on initial prices, a significant complication. The third element mandates that markets be proven to be endogenous complete. Our proof of existence overcomes these two hurdles 

Finally, it also relates to the general literature on existence and uniqueness of Arrow-Debreu equilibria in the standard model with homogeneous information, notably \cite{hugonnier2012endogenous}, \cite{MR3583456} and \cite{MR3131287} for dynamic models. Original contributions in that area focused on static economies and can be found in \cite{arrow1954} and \cite{mckenzie1959}.

\nada{
1. Literature on existence of DNREE in finite horizon models:\\
DRR (2020). Existence of PCE with bounded rational noise trading; existence of DNREE under specific assumptions; general diffusion dynamics. \\

H. He and J. Wang (1995). Differential Information and Dynamic Behavior of Stock Trading Volume. Review of Financial Studies, 8(4), 1995, 919–972. Model with exogenous noise trading; Gaussian factors; existence is not proved.
}

The paper is organized as follows. Section \ref{S:setup} presents the model and main assumptions. Section \ref{S:gen_U_arb_w} states the optimal investment problems and defines notions of equilibrium. Section \ref{S:MR} states the main result, and formally derives the equilibrium clearing condition central to the analysis. Section \ref{S:clear_solve} constructs  solutions to the market clearing condition at the signal realization level. Section \ref{S:PCE} establishes the existence of a PCE and Section \ref{S:DNREE} shows the PCE is a DNREE. Section \ref{S:asympt} identifies prices when $U$ has CRRA utility as the risk aversion tends to zero and infinity.  Section \ref{S:EX} contains an example when the terminal payoff is that of a geometric Brownian motion. Appendix \ref{AS:init_enlarge} contains results on initially enlarged filtrations from \cite[Online Appendix]{MR4192554} for ease of reference. Short proofs are in the main text, while longer proofs  are in Appendices \ref{AS:gen_U_arb_w} -- \ref{AS:asympt}.

\section{Setup}\label{S:setup}

\subsection*{The factor process and terminal payoff} Fix an integer $d$. The underlying probability space $\probtriple$ supports a $d$-dimensional Brownian motion $B$, and we denote by $\filt^B$ the $\prob$-augmentation of $B$'s natural filtration. The underlying factor process $X$ is a time homogeneous diffusion taking values in $\reals^d$, and satisfying the SDE
\begin{equation*}
dX_t = b(X_t)dt + a(X_t)dB_t, \qquad X_0 = x \in \reals^d,
\end{equation*}
where  we assume
\begin{ass}\label{A:X_SDE}
The drift function $b:\reals^d \to \reals^d$ is bounded and uniformly Lipschitz. The volatility function $a:\reals^d \to \reals^{d\times d}$ is bounded, uniformly Lipschitz, and invertible for each $x\in\reals^d$. Furthermore, the map $x\to a^{-1}(x)$ is bounded.
\end{ass}

\begin{rem} Our requirements ensure a unique strong solution $X$ taking values in $\reals^d$ for any starting point $x$. Also, the conditions on $b,a$ (as well the requirements in Assumption \ref{A:Psi} below) are  enforced to invoke the non-degeneracy results of \cite{MR3131287, MR3989959}. Beyond regularity, well posed-ness and local ellipticity, our requirements may be relaxed to the same extent as they could to use the conclusions of \cite{MR3131287, MR3989959}.
\end{rem}

The investment horizon is $[0,T]$ for $T>0$ fixed. There is a money market account with interest rate exogenously set to $0$. The risky assets have outstanding supply $\Pi \in\reals^d$ and terminal payoff $\Psi(X_T)$ for a given function $\Psi$, where we assume

\begin{ass}\label{A:Psi} The payoff function $\Psi: \reals^d \to\reals^d$ is $C^1$ and such that $\Pi'\Psi(x) \geq 0, x\in \reals^d$. Additionally, the Jacobian of $\Psi$ is of full rank for Lebesgue almost every $x\in \reals^d$ and there is $K>0$ such that both $|\Psi(x)| \leq e^{K(1+|x|)}$ and $|\Psi_i(x)| \leq e^{K(1+|x|)}, i=1,...,d$, where $\Psi_i(x) = \partial_{x_i} \Psi(x)$.
\end{ass}



\begin{rem}
Though not apparent at first, our assumptions allow $X$ to follow an OU process. Indeed, in terms of model primitives, $X$ is only used through its terminal value $X_T$, which appears both in the terminal payoff $\Psi(X_T)$ and the insider signal, defined below as $G_I = X_T + Y_I$. Thus, we can switch from $X$ to $\wt{X}$ where $\wt{X}_t \dfn \condexpvs{X_T}{\F^B_t}, t\leq T$, and calculation shows that if $dX_t = \kappa(\theta-X_t)dt + \sigma dB_t$ then $d\wt{X}_t = e^{-(T-t)\kappa}\sigma dB_t$. Though $\wt{X}$ has time dependent volatility, the volatility function satisfies the analyticity assumptions in \cite{MR3131287} and regularity assumptions in \cite{MR3989959}, and hence results go through using $\wt{X}$ instead of $X$.

We may also handle when the terminal payoff is that of a geometric Brownian motion. Indeed, this corresponds to $dX_t = \mu_X dt + \sigma_X dB_t$ for a certain vector $\mu_X$ and matrix $\sigma_X$,\footnote{In other words, the insider's signal is multiplicative, rather than additive, in the terminal stock price.} and $\Psi(x) = e^{x}$. This specification is considered in Section \ref{S:EX} below.
\end{rem}

\subsection*{Agents and signals}

At time $0$, there is an insider, or informed agent ``$I$'' who observes the private signal
\begin{equation*}
G_I = X_T + Y_I;\qquad  Y_I \sim N(0,C_I), \ Y_I \indep \F^B_T.
\end{equation*}
There is also a uninformed agent ``$U$'' who possesses no private signal.  Lastly, there is a misinformed agent, or noise trader ``$N$'' who thinks he receives $G_I$ at $0$, but in actuality receives
\begin{equation*}
G_N = \tau_N G_I + \mu_N + Y_N;\qquad Y_N \sim N(0,C_N), \ Y_N\indep Y_I,\F^B_T.
\end{equation*}
Above, $\mu_N\in\reals^d$ and $C_I,C_N\in\mathbb{S}^d_{++}$ the set of strictly positive definite symmetric $d\times d$ matrices. While we allow for general $\tau_N,\mu_N,C_N$, our noise trader specification is primarily made to incorporate two cases of interest.  First, when $\tau_N=1,\mu_N=0$. This corresponds to the noise trader mis-perceiving the informed signal's precision. Second, when $\tau_N = 0, \mu_N = E[X_T], C_N = \textrm{Var}[X_T] + C_I$. This corresponds to when the noise trader receives a truly independent signal, but one with the same first two moments as $G_I$. In either case, the noise trader fails to understand the true nature of the signal received, therefore trades on noise. We use the term ``noise trader'' to conform to the standard terminology, in line with \cite{black1986noise}, but as will be discussed next, $N$ is in many ways a rational agent. He only deviates from rationality, or is ``boundedly rational'' in that he believes his signal $G_N$ is correct, when in fact it is not.

Agents $I$ and $N$ have CARA, or exponential, utilities with respective risk aversions $\gamma_I,\gamma_N$.\footnote{CARA utility corresponds to the utility function $- e^{-\gamma w}, w\in\reals$ so that the absolute risk aversion is $\gamma$.} As is standard in the literature, to justify the price taking assumption that agents do not internalize their impact on prices, we assume agents are representative of a group of identical, ``small'' agents and that collectively, the insider and noise traders have  respective weights in the economy of $\omega_I,\omega_N$. The uninformed agent $U$ has weight in the economy of $\omega_U = 1-\omega_I- \omega_N$ and preferences dictated by a utility function $\util$\footnote{We use $\util$, as opposed to the more common $U$, as $U$ is the uninformed agent label.} for which we make the standard assumptions, as in \cite{MR1722287}.

\begin{ass}\label{A:U}
The utility function $\util:(0,\infty) \to \reals$ is $C^2$, strictly increasing, strictly concave, satisfies the Inada conditions $\lim_{w\to 0} \dot{\util}(w) = \infty, \lim_{w\to\infty} \dot{\util}(w) = 0$, as well as the asymptotic elasticity condition $\limsup_{w\to\infty} w\ \dot{\util}(w)/\util(w) < 1$.
\end{ass}

\begin{rem}
Our results easily extend to when the domain of $\util$ is $(a,\infty)$ for $a>-\infty$. We take $a=0$ for notational ease.   While our assumptions do not cover CARA preferences for $\util$, this case was extensively studied  in \cite{MR4192554}.
\end{rem}

Agent $j \in\cbra{I,N,U}$  is endowed at time $0$ with a constant $\pi^j_0$ shares in the risky asset, and no position in the money market.  The endowed initial positions satisfy

\begin{ass}\label{A:init_pos}
Share endowments $\{\pi^j_0\}, j\in\cbra{I,U,N}$ are consistent with market clearing in that
\begin{equation*}
\Pi = \omega_I \pi^{I}_0 + \omega_N \pi^N_0 +\omega_U \pi^{U}_0,
\end{equation*}
and the uninformed agent is allowed not to trade in that $(\pi^U_0)'\Psi(x) > 0$ for all $x\in\reals^d$.
\end{ass}

Agents maximize utility from terminal wealth given their information, and other than the share position, have no endowments.




\section{Equilibrium definition}\label{S:gen_U_arb_w}

We now define the notions of partial communication equilibrium (PCE) and dynamic noisy rational expectations equilibrium (DNREE), starting with the PCE.

\subsection*{Information flows and price process} In the PCE, there is a public  filtration $\filt^m$ enlarging $\filt^B$ and satisfying the usual conditions.   This is the filtration used by $U$, and as we will show, it takes the form $\filt^m = \filt^B \vee \sigma(H)$ for a  public signal $H$ which, as with the equilibrium price process, must be found ($H$ is explicitly given in \eqref{E:H_def} below). The idea behind $\filt^m$ initially enlarging $\filt^B$ is that a portion of the insider's private information signal is communicated at the outset to the uninformed agent, either directly through communication, or indirectly through observed endogenous quantities. Intuition for this structure of the market filtration stems from the fact that  $I$ and $N$ have CARA preferences. The insider, as she uses public information as well as her private signal, has filtration $\filt^I = \filt^m \vee \sigma(G_I)$, the initial enlargement of $\filt^m$ with respect to the private signal $G_I$, or equivalently, the initial enlargement of $\filt^B$ with respect to the signal pair $(H,G_I)$.

Given the public information flow $\filt^m$, the (to-be-determined) price process $S$ is a $\filt^m$ semi-martingale satisfying the terminal condition $S_T = \Psi(X_T)$.

\subsection*{Individual optimal investment problems}

Let us start with $U$ who solves the optimal investment problem
\begin{equation}\label{E:opt_U_signal}
\sup_{\pi\in\A_U} \condexpvs{\util\left((\pi^U_0)'S_0 + \We^{\pi}_{0,T}\right)}{\F^m_0},
\end{equation}
where the gains process $\We^{\pi}$ is defined by
\begin{equation}\label{E:gains_process}
    \We^{\pi}_{0,T} \dfn \int_0^T \pi_u'dS_u,
\end{equation}
The class of admissible strategies $\A_U$ consists of $\filt^m$ predictable, $S$ integrable processes $\pi$ satisfying a certain budget constraint. We postpone explicitly identifying $\A_U$ until the end of the section, as its construction is rather delicate, for now taking $\A_U$ as given.  Similarly, and because the initial wealth factors out of the conditional expectation, $I$ solves the investment problem
\begin{equation}\label{E:opt_I_signal}
\inf_{\pi\in\A_I} \condexpvs{e^{-\gamma_I \We^{\pi}_{0,T}}}{\F^I_0},
\end{equation}
where $\A_I$, also explicitly constructed below consists of $\filt^I$ predictable, $S$ integrable processes with a related budget constraint.  Implicit in this definition is that $S$ is a $\filt^I$ semi-martingale, a fact we will show at the end of this section as well.

Provided they exist, denote by  $\wh{\pi}^U,\wh{\pi}^I$ the optimal policies for $U$ and $I$ respectively. Regarding the noise trader $N$, as proved in \cite[Remark 2.2 and Section S.3]{MR4192554}  our noise trader convention  implies that if  $I$ has optimal strategy $\wh{\pi}^I_t = (1/\gamma_I) \wh{\psi}^{G_I}_t$ where $\wh{\psi}^g_t = \wh{\psi}(t,\omega,g)$ is $\mathcal{P}(\filt^m)\times \mathcal{B}(\reals^d)$ measurable,\footnote{$\mathcal{P}$ is the predictable sigma field and $\mathcal{B}$ are the Borel sets.}  then $\wh{\pi}^N_t = (1/\gamma_N)\wh{\psi}^{G_N}_t$. In words, modulo risk aversion the noise trader  uses the same ``signal to strategy'' function as the insider. However, he plugs an incorrect signal into the function. This means he acts like an informed agent with characteristics $\left(\omega_N,\pi^N_0,\gamma_N\right)$, processing $G_N$ as if it were the true signal $G_I$.

\subsection*{Equilibrium definition}

With this notation, we define the PCE in terms of the price process $S$ and public filtration $\filt^m = \filt^B \vee \sigma(H)$ as follows
\begin{defn}
The pair $(S,\filt^m)$ is a Partial Communication Equilibrium (PCE) if $\filt^B \subset \filt^m \subset \filt^I$, $S_T = \Psi(X_T)$ almost surely, and
\begin{equation}\label{E:eq_cond}
\Pi = \omega_I \wh{\pi}^I + \omega_N \wh{\pi}^N  + \omega_U \wh{\pi}^U ;\qquad \textrm{Leb}_{[0,T]}\times \prob \textrm{  almost surely}.
\end{equation}
\end{defn}
The idea behind the strict informational set inclusions is that, in a PCE, the market is provided information in addition to that given by the factor process, but the insider still has the finest information set.  This rules out the unrealistic (though possible - see \cite[Section 3]{MR4192554}) case when the insider's signal is passed to the market so $\filt^m = \filt^I = \filt^B \vee \sigma(G_I)$.  Note also that in \eqref{E:eq_cond} the strategies depend upon $S, \filt^m$ and $G_I,G_N$. We next turn to the DNREE.
\begin{defn}\label{D:DNREE}
The pair $(S,\filt^m)$ is a Dynamic Noisy Rational Expectations Equilibrium (DNREE) if $(S,\filt^m)$ is a PCE and if $\filt^m = \filt^{B,S}_{+}$, where $\filt^{B,S}_{+}$ is the $\prob$-augmented right-continuous enlargement of the filtration generated by $B$ and $S$.
\end{defn}
We see that a DNREE prevails if the market signal is  recoverable at time $0+$ (i.e. for any time $t>0$) based upon market observable quantities. This provides a mechanism for transmission of the market signal. Note also that by right-continuity of $\filt^m$  we know that $\filt^{B,S}_{+} \subseteq \filt^m$ automatically. Thus, under a DNREE we have $\filt^m = \filt^{B,S}_{+} \subset \filt^I$. Lastly, we must use the right-continuous enlargement of $\filt^{B,S}$ because it might not be true that either $\filt^{B,S}$ is right continuous or (more importantly) that the public signal $H$  is $\sigma(S_0)$ measurable. For that reason, the notion of a DNREE defined above is weaker than in \cite{MR4192554}.

\subsection*{The acceptable trading strategies}

We conclude by constructing the classes of acceptable trading strategies. While the construction is technical, informally, the acceptable strategies are the standard ones encountered in a complete market optimal investment problem. In other words, if we assume the $\filt^m$ semi-martingale $S$ is also a $\filt^I$ semi-martingale, and if we denote by $L^m(S)$ (respectively $L^I(S)$) the set of $\filt^m$ (resp. $\filt^I$) predictable, $S$-integrable strategies then\footnote{All inequalities between random variables are assumed to hold almost surely.}
\begin{equation}\label{E:init_accept}
    \begin{split}
        \A_U &\approx \wh{\A}_U \dfn \cbra{\pi \in L^m(S) \such \condexpv{\qprob^m}{}{\We^{\pi}_{0,T}}{\F^m_0} \leq 0},\\
        \A_I &\approx \wh{\A}_{I} \dfn \cbra{\pi \in L^I(S) \such  \condexpv{\qprob^I}{}{\We^{\pi}_{0,T}}{\F^I_0} \leq 0},
    \end{split}
\end{equation}
where $\qprob^m$ and $\qprob^I$ are the unique martingale measures in the $(S,\filt^m)$ and $(S,\filt^I)$ markets respectively. With these sets, the agents then solve the investment problems \eqref{E:opt_U_signal}, \eqref{E:opt_I_signal} and equilibrium will follow provided the optimal policies clear the market.

We cannot directly use $\wh{\A}_U$ and $\wh{\A}_I$ due to the prohibitive difficulty in establishing (for example) that the optimal gains $\We^{\wh{\pi}^U}_{0,T}$ for the uninformed agent is in $L^1(\qprob^m)$, as we have very little control over the map $H \to \We^{\wh{\pi}^U}(H)_{0,T}$. However, from $U$'s perspective $H$ is ``constant'', and thus (provided the requisite measurability)  it is more natural to only require integrability of $\We^{\wh{\pi}^U}(h)_{0,T}$ for each signal realization $h$ (with similar statements for the insider $I$). Our construction of the acceptable trading strategies builds upon this idea, using results, summarized in Appendix \ref{AS:init_enlarge} below, from the theory of initially enlarged filtrations, generalized conditional expectations, and the (expected) sample path continuity of the equilibrium price process, which enables us to avoid many technicalities which arise when the price process may jump. Below we carry out the construction, with all proofs within this section given in Appendix \ref{AS:gen_U_arb_w}.


To begin, let $H$ be a $\reals^d$ valued random variable, possessing a strictly positive probability density function (pdf) on $\reals^d$,\footnote{For the general theory, it is not necessary for $H$ to possess a pdf, however, $H$ from \eqref{E:H_def} does possess a pdf, and the presentation of this section simplifies under this assumption.} satisfying the Jacod equivalence condition (c.f. \cite{Jacod1985}) $\condprobs{H\in dh}{\F^B_t} \sim \prob\bra{H\in dh}$ almost surely on $t\leq T$ , and define $\filt^m \dfn \filt^B \vee(H)$. From Lemma \ref{L:fontana_facts} in Appendix \ref{AS:init_enlarge} we know $\filt^m$ satisfies the usual conditions, and if we define the density process
\begin{equation}\label{E:ph_def}
p^h_t \dfn \frac{d\condprobs{H \in dh}{\F^B_t}}{d\prob\bra{H\in dh}},\qquad t\leq T,
\end{equation}
then $(t,\omega,h) \to p^h_t(\omega)$ is a $\OO(\filt^B)\otimes \B(\reals^d)$ measurable map. Using $p^H$, we define the $\filt^B$ to $\filt^m$ martingale preserving measure (see \cite{MR1418248, MR1632213,MR1775229}) $\tprob^H$ on $\F^m_T$ by
\begin{equation}\label{E:tprob_def}
\frac{d\tprob^H}{d\prob}\big|_{\F^m_T} \dfn \frac{1}{p^H_T}.
\end{equation}
To construct the price process $S$, let $(\omega,h)\to \cz^h(\omega)$ be a strictly positive $\F^B_T\otimes \B(\reals^d)$ measurable random variable, and assume for Lebesgue almost every (``Leb a.e.'') $h\in\reals^d$ that $\expvs{\cz^h} = 1$ and $\expvs{|\Psi(X_T)|\cz^h} < \infty$.  Define the measures $\qprob^h$ on $\F^B_T$ and $\qprob^m$ on $\F^m_T$ by
\begin{equation*}
    \frac{d\qprob^h}{d\prob}\big|_{\F^B_T} \dfn \cz^h;\qquad \frac{d\qprob^m}{d\prob}\big|_{\F^m_T} \dfn \frac{\cz^H}{p^H_T}.
\end{equation*}
Lemma \ref{L:market_filt_cond_exp} in Appendix \ref{AS:init_enlarge} shows that for any $\F^B_t\otimes\B(\reals^d)$ measurable random variable $\chi^H_t$ such that $\expv{\qprob^h}{}{|\chi^h_t|}<\infty$ for Leb a.e. $h\in\reals^d$, the conditional expectation $\condexpv{\qprob^m}{}{\chi^H_t}{\F^m_s}$ is well defined (c.f. the discussion right above Lemma \ref{L:market_filt_cond_exp}) and that
\begin{equation}\label{E:qprobm_to_qprobh}
    \condexpv{\qprob^m}{}{\chi^H_t}{\F^m_s} = \left(\condexpv{\qprob^h}{}{\chi^h_t}{\F^B_s}\right)\big|_{h=H},
\end{equation}
which allows us to connect $(\qprob^m,\filt^m)$ and $(\qprob^h,\filt^B)$ conditional expectations. Given this, define the price process $S = S^H$ by
\begin{equation}\label{E:price_def}
S^H_t \dfn \condexpv{\qprob^m}{}{\Psi(X_T)}{\F^m_t} = \left(\condexpv{\qprob^h}{}{\psi(X_T)}{\F^B_t}\right)\big|_{h=H};\quad t \leq T,
\end{equation}
so that $S^H$ is a RCLL $\filt^m$ semi-martingale, and hence by \cite[Lemma 4.3]{MR3758346}, is indistinguishable  from an $\OO(\filt^B)\otimes \B(\reals^d)$ measurable process (which we use going forward). Clearly this latter process is the map $(t,\omega,h) \to S^h_t(\omega) \dfn \condexpv{\qprob^h}{}{\psi(X_T)}{\F^B_t}$, and thus $S^h$ is a continuous $(\filt^B)$ semi-martingale. Additionally, we have

\begin{lem}\label{L:S_h_prop}
For $h\in\reals^d$, $S^h$ has quadratic covariation $\langle S^h, S^h\rangle_t = \int_0^t \Sigma^h_u du$ where the map $(t,\omega,h) \to \Sigma^h_t(\omega)$ is $\mcp(\filt^B)\otimes \B(\reals^d)$ measurable.
\end{lem}

As a last preparatory step, for fixed $h\in\reals^d$ define the set of processes  (here ``$X\in \Hcal$'' means that $X$ is $\Hcal$ measurable)
\begin{equation*}
    \A^h \dfn \cbra{ \pi\in\mcp(\filt^B) \such \pi \textrm{ is } S^h \textrm{ integrable, } \int_0^T \pi_t'dS^h_t \in L^1(\qprob^h) \textrm{ with } \expv{\qprob^h}{}{\int_0^T \pi_t'dS^h_t} \leq 0}.
\end{equation*}
With all of this in place, the set of acceptable trading strategies for the uninformed agent is
\begin{equation}\label{E:AU}
    \begin{split}
        \A_U &\dfn \cbra{\pi \in \mcp(\filt^B)\otimes \B(\reals^d) \such \pi^h \in \A^h \textrm{ for } \textrm{ Leb a.e. } h\in\reals^d }
    \end{split}
\end{equation}
To connect $\A_U$ with $\wh{\A}_U$ from \eqref{E:init_accept} we have the following Proposition.
\begin{prop}\label{P:class_equiv}\text{}
    \begin{enumerate}[(i)]
    \item $\A_U \subseteq \wh{\A}_U$ and for each $\wh{\pi}\in\wh{\A}_U$ there is a $\pi\in\A_U$ such that the processes $\We^{\pi}_{0,\cdot}$, $\We^{\wh{\pi}}_{0,\cdot}$ are indistinguishable on $[0,T]$.
    \item For $\pi\in\A_U$ the processes $\We^{\pi}_{0,\cdot}$ and $\left(\int_0^\cdot (\pi^h_u)'dS^h_u\right)\big|_{h=H}$ are  both $\mcp(\filt^B)\otimes \B(\reals^d)$ measurable and indistinguishable.
    \item For $\pi\in\A_U$, $\condexpv{\qprob^m}{}{\We^{\pi}_{0,T}}{\F^m_0}$ is well defined, and  $\condexpv{\qprob^m}{}{\We^{\pi}_{0,T}}{\F^m_0} = \left(\expv{\qprob^h}{}{\int_0^T (\pi^h_u)'dS^h_u}\right)\big|_{h=H}$.
    \end{enumerate}
\end{prop}

We next turn to the insider. First, assume $H$ and $G_I$ are also such that the Jacod condition $\condprobs{G_I \in dg}{\F^m_t} \sim \prob\bra{G_I \in dg}$ holds almost surely on $t\leq T$  and write the density as
\begin{equation*}
p^{H, g}_t \dfn \frac{d\condprobs{G_I \in dg}{\F^m_t}}{d\prob\bra{G_I\in dg}},\qquad t\leq T.
\end{equation*}
Lemma \ref{L:fontana_facts} in Appendix \ref{AS:init_enlarge} shows $\filt^I$ satisfies the usual conditions and the map $(t,\omega,g) \to p^{H(\omega),g}_t(\omega)$ is $\OO(\filt^m)\otimes\B(\reals^d)$ measurable. Furthermore, for $H$ in \eqref{E:H_def},  \eqref{E:H_to_G} below implies $(t,\omega,h,g) \to p^{h,g}_t(\omega)$ is $\mcp(\filt^B)\otimes \B(\reals^d \times \reals^d)$ measurable. Using standard results on the martingale preserving measure  one can show  $S^H$ is a $(\qprob^I,\filt^I)$ martingale where $\qprob^I$ is defined by
\begin{equation}\label{E:first_QI_def}
\frac{d\qprob^I}{d\prob}\big|_{\F^I_T} = \frac{1}{p^{H,G_I}_T} \times \frac{d\qprob^m}{d\prob}\big|_{\F^m_T}.
\end{equation}
Indeed, mimicking the proof of \cite[Theorem 3.1]{MR1775229}, if $M$ is a $(\qprob^m,\filt^m)$ martingale then for $0\leq s < t\leq T$, $A^m_s \in \F^m_s$ and $\psi$ a smooth bounded function on $\reals^d$ we have
\begin{equation*}
    \begin{split}
        \expv{\qprob^I}{}{M_t 1_{A^m_s}\psi(G_I)} &= \expvs{\frac{1}{p^{H,G_I}_T} \frac{d\qprob^m}{d\prob}\big|_{\F^m_T} M_t 1_{A^m_s}\psi(G_I)} = \expvs{\frac{d\qprob^m}{d\prob}\big|_{\F^m_T} M_t 1_{A^m_s}\condexpvs{\frac{1}{p^{H,G_I}_T} \psi(G_I)}{\F^m_T}}.
    \end{split}
\end{equation*}
By construction $\condexpvs{\psi(G_I)/p^{H,G_I}_T}{\F^m_T} = \int_{\reals^d} \psi(g)\prob\bra{G_I\in dg} \rdfn K(\psi)$. This implies
\begin{equation*}
    \begin{split}
        \expv{\qprob^I}{}{M_t 1_{A^m_s}\psi(G_I)} &= K(\psi)\expvs{\frac{d\qprob^m}{d\prob}\big|_{\F^m_T} M_t 1_{A^m_s}} = K(\psi)\expv{\qprob^m}{}{M_t 1_{A^m_s}}  = K(\psi)\expv{\qprob^m}{}{M_s 1_{A^m_s}}.
    \end{split}
\end{equation*}
Letting $t=s$ above proves the martingale property. Given this,  similarly to $\A_U$ we define for $I$
\begin{equation}\label{E:AI}
    \begin{split}
        \A_I &\dfn \cbra{\pi \in \mcp(\filt^B)\otimes \B(\reals^{2d}) \such  \pi^{h,g} \in \A^{h} \textrm{ for } \textrm{ Leb a.e. } (h,g)\in\reals^{2d}},
    \end{split}
\end{equation}
and the analog of Proposition  \ref{P:class_equiv} is the following, the proof of which is identical to that of Proposition \ref{P:class_equiv} and hence omitted.
\begin{prop}\label{P:I_class_equiv}\text{}
    \begin{enumerate}[(i)]
    \item $\A_I \subseteq \wh{\A}_I$ and for each $\wh{\pi}\in\wh{\A}_I$ there is a $\pi\in\A_I$ such that the processes $\We^{\pi}_{0,\cdot}$, $\We^{\wh{\pi}}_{0,\cdot}$ are indistinguishable on $[0,T]$.
    \item For $\pi\in\A_I$ the processes $\We^{\pi}_{0,\cdot}$ and $\left(\int_0^\cdot (\pi^{h,g}_u)'dS^h_u\right)\big|_{h=H,g=G_I}$ are  both $\mcp(\filt^B)\otimes \B(\reals^{2d})$ measurable and indistinguishable.
    \item For $\pi\in\A_I$, $\condexpv{\qprob^I}{}{\We^{\pi}_{0,T}}{\F^I_0}$ is well defined, and  $\condexpv{\qprob^I}{}{\We^{\pi}_{0,T}}{\F^I_0} = \left(\expv{\qprob^h}{}{\int_0^T (\pi^{h,g}_u)'dS^h_u}\right)\big|_{h=H,g=G_I}$.
    \end{enumerate}
\end{prop}

\section{The Main Result}\label{S:MR}

With individual optimal investment problems and equilibrium notions defined, our main result is the following theorem. To state it,  define the insider and noise trader weighted risk tolerances
\begin{equation}\label{E:weighted_risk_tol_IN}
\alpha_j \dfn \frac{\omega_j}{\gamma_j};\quad j\in {I,N}.
\end{equation}

\begin{thm}\label{T:main_result}
Let Assumptions \ref{A:X_SDE}, \ref{A:Psi}, \ref{A:U} and \ref{A:init_pos} hold.  Then, there exists a PCE.  The market signal is
\begin{equation}\label{E:H_def}
H \dfn \frac{\alpha_I G_I + \alpha_N (G_N-\mu_N)}{\alpha_I + \alpha_N\tau_N} = X_T + Y_I + \frac{\alpha_N}{\alpha_I + \alpha_N\tau_N}Y_N.
\end{equation}
The $(S,\filt^m)$ and $(S,\filt^I)$ markets are complete. Lastly, in the univariate case, $d=1$, the PCE is a DNREE.
\end{thm}

\begin{rem}
Given the problem setup in Section \ref{S:gen_U_arb_w}, identifying  $(S,\filt^m)$ which enforce a PCE is equivalent to identifying the market signal $H$ (which determines $\filt^m)$ and the terminal density $\cz^h$ (which along with $\filt^m$ identifies $S$).  As such, one may also interpret Theorem \ref{T:main_result} as saying there exists $H,\cz^h$ such that a PCE holds. However, to conform to the standard presentation for equilibrium results in terms of prices and information, we choose to express the PCE in terms of $(S,\filt^m)$.
\end{rem}

The proof of Theorem \ref{T:main_result} is carried out in the Sections \ref{S:clear_solve} -- \ref{S:DNREE} below.  In the remainder of this section, we formally derive a clearing condition which, provided markets are complete, identifies the measure $\qprob^m$. In the later sections we will rigorously establish the result.

\subsection*{The clearing condition} The signal $H$ of \eqref{E:H_def} is of the same form as $G_I$, but with  lower precision
\begin{equation*}
P_U \dfn C_U^{-1};\qquad C_U \dfn C_I + \left(\frac{\alpha_N}{\alpha_I + \alpha_N\tau_N}\right)^2 C_N.
\end{equation*}
With $p_C$ denoting the probability density function (pdf) of a $N(0,C)$ random vector, direct calculations show $H$ has  $\F^B_t$ conditional pdf $\ell(t,X_t,h)$ where
\begin{equation}\label{E:ell_T_def}
\ell(t,x,h) \dfn \condexpvs{p_{C_U}(h-X_T)}{X_t = x};\quad \ell(T,x,h) = p_{C_U}(h) e^{-\frac{1}{2}x'P_U x + x'P_Uh},
\end{equation}
so that $\condprobs{H\in dh}{\F^B_t} \sim \prob\bra{H\in dh}$ almost surely on $t\leq T$, with density $p^h_t = \ell(t,X_t,h)/\ell(0,X_0,h)$. To construct the market martingale measure $\qprob^m$, we define the  $\filt^B$ to $\filt^m$ martingale preserving measure $\tprob^H$ as in \eqref{E:tprob_def} and write $\wt{\mathbb{E}}^H$ for expectations with respect to $\tprob^H$.  Slightly changing notation from Section \ref{S:gen_U_arb_w}, and motivated by our Markovian environment, let $\cz$ be a strictly positive $\B(\reals^{2d})$ measurable function such that for Leb a.e. $h\in\reals^d$ we have both $\expvs{\cz(X_T,h)} < \infty$  and $\expvs{|\Psi(X_T)|\cz(X_t,h)} < \infty$.\footnote{$\cz^h$ of Section \ref{S:gen_U_arb_w} takes the form $\cz^h = \cz(X_T,h)/\expvs{\cz(X_T,h)}$.}  The measures $\qprob^h$ and $\qprob^m$ of Section \ref{S:gen_U_arb_w} are defined $\qprob^h$ through $d\qprob^h/d\prob|_{\F^B_T} = \cz(X_T,h)/\expvs{\cz(X_T,h)}$
\begin{equation}\label{E:qprob_m}
\frac{d\qprob^m}{d\prob}\big|_{\F^m_T} =\frac{\cz(X_T,H)}{\expvs{\cz(X_T,h)}|_{h=H}}\times \frac{1}{p^H_T} = \frac{\cz(X_T,H)}{\expvs{\cz(X_T,h)}|_{h=H}}\times \frac{\ell(0,X_0,H)}{\ell(T,X_T,H)}.
\end{equation}

With $\qprob^m$ so defined, we define the price process $S^H$ as in \eqref{E:price_def}, and $\qprob^m$ is a martingale measure, in the sense that $\condexpv{\qprob^m}{}{S^H_t}{\F^m_s}$ is well defined for each $s < t\leq T$ and almost surely equal to $S^H_s$.

The informed agent $I$ has filtration $\filt^I = \filt^m \vee \sigma(G_I)$, and has the $\filt^I$ martingale measure $\qprob^I$ associated to $\qprob^m$ from \eqref{E:first_QI_def}, with terminal density
\begin{equation*}
\frac{d\qprob^I}{d\prob}\big|_{\F^I_T} = \frac{\cz(X_T,H)}{\expvs{\cz(X_T,h)}|_{h=H}}\times\frac{1}{p^H_T p^{H,G_I}_T}.
\end{equation*}
Additionally, from Lemma \ref{L:market_filt_cond_exp} of Appendix \ref{AS:init_enlarge}, we deduce for all appropriately integrable, $\F^B_t\otimes \B(\reals^{2d})$ measurable, random variables $\chi^{H,G}_t$, that
\begin{equation}\label{E:exp1}
\condexpv{\qprob^I}{}{\chi^{H,G_I}_t}{\F^I_s} = \left(\condexpv{\qprob^m}{}{\chi^{H,g}_t}{\F^m_s}\right)\big|_{g=G_I} = \left(\condexpv{\qprob^h}{}{\chi^{h,g}_t}{\F^m_s}\right)\big|_{h=H,g=G_I}.
\end{equation}
To simplify $p^H_T p^{H,G_I}_T$, write $u^g_t$ as the $\F^B_t$ conditional pdf of $G_I$. Calculation shows $u^g_t = u(t,X_t,g)$ where
\begin{equation*}
    u(t,x,g) = \condexpvs{p_{C_I}(g-X_T)}{X_t = x};\quad u(T,x,g) = p_{C_I}(g) e^{-\frac{1}{2}x'P_I x + x'P_I g},
\end{equation*}
where above (and for later use)
\begin{equation*}
    P_I \dfn C_I^{-1};\qquad P_N \dfn C_N^{-1}.
\end{equation*}
With this notation, direct calculations show
\begin{equation}\label{E:H_to_G}
p^H_T p^{H,G_I}_T = \delta(H,G_I) \times p_{C_I}(G_I-X_T),
\end{equation}
where $\delta(H,G_I)$ is an explicitly identifiable function, and hence
\begin{equation*}
    \frac{d\qprob^I}{d\prob}\big|_{\F^I_T} =  \delta(H,G_I)\times \cz(X_T,H)e^{\frac{1}{2}X_T'P_I X_T - X_T'P_I G_I}
\end{equation*}
where the function $\delta$ changed from the previous equality. With this notation, the first order optimality condition for $I$ is
\begin{equation}\label{E:opt_I}
\begin{split}
-\gamma_I \wh{\We}^{I}_{0,T} 
&= \log(\cz(X_T,H)) + \frac{1}{2}X_T'P_IX_T - \condexpv{\qprob^m}{}{\log(\cz(X_T,H)) + \frac{1}{2}X_T'P_IX_T}{\F^m_0}\\
&\qquad - G_I'P_I\left(X_T - \condexpv{\qprob^m}{}{X_T}{\F^m_0}\right),
\end{split}
\end{equation}
where we used \eqref{E:exp1}. Therefore, by our noise trader convention
\begin{equation}\label{E:opt_N}
\begin{split}
-\gamma_N \wh{\We}^{N}_{0,T} &= \log(\cz(X_T,H)) + \frac{1}{2}X_T'P_IX_T - \condexpv{\qprob^m}{}{\log(\cz(X_T,H)) + \frac{1}{2}X_T'P_IX_T}{\F^m_0}\\
&\qquad \qquad  - G_N'P_I\left(X_T - \condexpv{\qprob^m}{}{X_T}{\F^m_0}\right).
\end{split}
\end{equation}
Using \eqref{E:weighted_risk_tol_IN} this gives the population weighted terminal gains $\omega_I \wh{\We}^I_{0,T} + \omega_N \wh{\We}^N_{0,T}$
\begin{equation*}
\begin{split}
&\omega_I\wh{\We}^I_{0,T} +  \omega_N \wh{\We}^N_{0,T} = \left(\left(\alpha_I + \alpha_N\tau_N\right)H +\alpha_N\mu_N\right)'P_I\left(X_T - \condexpv{\qprob^m}{}{X_T}{\F^m_0}\right)\\
&\qquad  -(\alpha_I + \alpha_N)\left(\log(\cz(X_T,H)) + \frac{1}{2}X_T'P_IX_T - \condexpv{\qprob^m}{}{\log(\cz(X_T,H)) + \frac{1}{2}X_T'P_IX_T}{\F^m_0}\right).
\end{split}
\end{equation*}
\begin{rem}
We pause briefly to discuss why the market signal $H$ takes the form in \eqref{E:H_def}.  Indeed, \eqref{E:opt_I}, \eqref{E:opt_N} imply the only component dependent on private information  in the population weighted gains of $I$ and $N$ is
\begin{equation*}
    \left(\alpha_I G_I + \alpha_N G_N\right)'P_I\left(X_T - \condexpv{\qprob^m}{}{X_T}{\F^m_0}\right).
\end{equation*}
We need this quantity to be $\F^m_T$ measurable.  Clearly this can be achieved by taking the market signal to be $H_{old} = (\alpha_I G_I + \alpha_N G_N)/(\alpha_I + \alpha_N)$, and this is what was done in \cite{MR4192554}. Presently we take $H$ as in \eqref{E:H_def}.  As $\sigma(H) = \sigma(H_{old})$ the equilibrium is unchanged choosing $H$ or $H_{old}$.  We choose $H$ as it has the same form as $G_I$ and hence yields the interpretation that the market gets a signal of the same form as the insider, just of a lower precision.
\end{rem}
We lastly consider $U$, who using \eqref{E:ell_T_def}, \eqref{E:qprob_m} has first order optimality condition
\begin{equation*}
\dot{\util}\left(\wh{\We}^U_T \right) =  \lambda(H) \check{Z}(X_T,H)e^{\frac{1}{2}X_T'P_U X_T - X_T'P_U H}, \qquad \wh{\We}^U_T \dfn (\pi^U_0)'S_0 + \wh{\We}^U_{0,T},
\end{equation*}
where $\lambda(H)$ is $\F^m_0$ measurable.  We can write this as
\begin{equation}\label{E:U_foc_I_def}
\wh{\We}^U_T = \imutil\left(\lambda(H) \check{Z}(X_T,H)e^{\frac{1}{2}X_T'P_U X_T - X_T'P_U H}\right),\qquad \imutil \dfn \dot{\util}^{-1},
\end{equation}
provided $\lambda$ enforces the static budget constraint
\begin{equation}\label{E:U_budget}
w^U_0 \dfn (\pi^U_0)'S_0 = \condexpv{\qprob^m}{}{\wh{\We}^U_T}{\F^m_0} = \condexpv{\qprob^m}{}{\imutil\left(\lambda(H) \check{Z}(X_T,H)e^{\frac{1}{2}X_T'P_U X_T - X_T'P_U H}\right)}{\F^m_0}.
\end{equation}
In equilibrium,  according to \eqref{E:eq_cond}, we must have market clearing
\begin{equation*}
\begin{split}
\Pi'(\Psi(X_T)-S_0) &= \omega_I \wh{\We}^I_{0,T} +  \omega_N \wh{\We}^N_{0,T} +  \omega_U\left(\wh{\We}^U_{T} - w^U_0\right).
\end{split}
\end{equation*}
Plugging in for the optimal wealth processes, we need $\check{Z}(X_T,H)$ to enforce
\begin{equation*}
\begin{split}
&\Pi'(\Psi(X_T) - S_0) = \omega_U\left( \imutil\left(\lambda(H) \check{Z}(X_T,H)e^{\frac{1}{2}X_T'P_U X_T - X_T'P_U H}\right) - w^U_0\right)\\
&\quad -(\alpha_I + \alpha_N)\left(\log(\cz(X_T,H)) + \frac{1}{2}X_T'P_IX_T\right)+ \left(\left(\alpha_I + \alpha_N\tau_N\right)H + \alpha_N \mu_N\right)'P_IX_T\\
&\quad -\condexpv{\qprob^m}{}{-(\alpha_I + \alpha_N)\left(\log(\cz(X_T,H)) + \frac{1}{2}X_T'P_IX_T\right)+ \left(\left(\alpha_I + \alpha_N\tau_N\right)H + \alpha_N \mu_N\right)'P_IX_T}{\F^m_0}.
\end{split}
\end{equation*}
To simplify this equation, define the function
\begin{equation}\label{E:chi_def}
\begin{split}
\chi(x,h) &\dfn \Pi'\Psi(x) + (\alpha_I + \alpha_N)\frac{1}{2}x'P_Ix - \left(\left(\alpha_I + \alpha_N\tau_N\right)h + \alpha_N \mu_N\right)'P_Ix.
\end{split}
\end{equation}
Using this and $S_0 = \condexpv{\qprob^m}{}{\Psi(X_T)}{\F^m_0}$ we obtain the market clearing condition
\begin{equation}\label{E:clear}
\begin{split}
&\chi(X_T,H) - \condexpv{\qprob^m}{}{\chi(X_T,H)}{\F^m_0} = \omega_U \left(\imutil\left(\lambda(H) \check{Z}(X_T,H)e^{\frac{1}{2}X_T'P_U X_T - X_T'P_U H}\right)- w^U_0\right)\\
&\qquad\qquad -(\alpha_I + \alpha_N)\left(\log(\cz(X_T,H))  - \condexpv{\qprob^m}{}{\log(\cz(X_T,H))}{\F^m_0}\right).
\end{split}
\end{equation}
If we can find functions $\cz(x,h),\lambda(h)$ solving \eqref{E:U_budget}, \eqref{E:clear} then, provided the resultant markets are complete and certain integrability conditions hold, the above wealth processes will be optimal and hence a PCE will exist.

\section{Solving the market clearing condition}\label{S:clear_solve}

In this section we identify solutions $\cz,\lambda$ to the market clearing condition. Throughout, Assumptions \ref{A:X_SDE}, \ref{A:Psi}, \ref{A:U} and \ref{A:init_pos} are in force. We first re-express $\cz,\lambda$ by letting $h\to \kappa(h)$ and $(x,h,\kappa) \to z(x,h,\kappa)\geq 0$ be (to-be-determined) measurable maps and setting
\begin{equation}\label{E:gen_cz_new}
    \begin{split}
    \lambda(h) &= \expvs{z(X_T,h,\kappa(h)) e^{-\frac{1}{2}X_T'P_U X_T + X_T'P_U h}},\\
    \cz(x,h) &= \frac{z(x,h,\kappa(h))e^{-\frac{1}{2}x'P_U x + x'P_U h}}{\expvs{z(X_T,h,\kappa(h)) e^{-\frac{1}{2}X_T'P_U X_T + X_T'P_U h}}},
    \end{split}
\end{equation}
where $\kappa,z$ are such that $\lambda(h)<\infty$ for each $h$. Using \eqref{E:price_def} and \eqref{E:qprobm_to_qprobh}, the budget constraint \eqref{E:U_budget} will hold if for each $h$,  writing $\kappa$ for $\kappa(h)$
\begin{equation}\label{E:U_budget_new}
    0 = \expv{\qprob^h}{}{(\pi^U_0)'\Psi(X_T)- \imutil(z(X_T,h,\kappa))};\qquad \frac{d\qprob^h}{d\prob}\big|_{\F^B_T} = \frac{z(X_T,h,\kappa)e^{-\frac{1}{2}x'P_U x + x'P_U h}}{\expvs{z(X_T,h,\kappa)e^{-\frac{1}{2}x'P_U x + x'P_U h}}}.
\end{equation}
Given the budget constraint, if we apply \eqref{E:qprobm_to_qprobh} again then \eqref{E:clear} will hold if
\begin{equation}\label{E:clear_real_2}
\begin{split}
\vartheta(x,h) - \expv{\qprob^h}{}{\vartheta(X_T,h)} &= \omega_U\left(\imutil(z(x,h,\kappa)) - \expv{\qprob^h}{}{\imutil(z(X_T,h,\kappa))}\right) \\
&\qquad -  (\alpha_I + \alpha_N)\big(\log(z(x,h,\kappa)) -\expv{\qprob^h}{}{\log(z(X_T,h,\kappa))}\big),
\end{split}
\end{equation}
where
\begin{equation}\label{E:vartheta_def}
\begin{split}
\vartheta(x,h) &\dfn \chi(x,h) - (\alpha_I + \alpha_N)\left(\frac{1}{2}x'P_U x - x'P_U h\right).
\end{split}
\end{equation}
Lastly, in \eqref{E:clear_real_2}, if we absorb the conditional expectations with respect to $\F^m_0 = \sigma(H)$  into $\kappa$, we seek solutions $z=z(x,h,\kappa)\geq 0$ to
\begin{equation}\label{E:clear_function_00}
    \begin{split}
    \omega_U \imutil(z) - (\alpha_I + \alpha_N)\log(z) &= \vartheta(x,h) + \kappa,
    \end{split}
\end{equation}
where $\kappa=\kappa(h)$ must enforce  \eqref{E:U_budget_new} and \eqref{E:clear_real_2}.

Our first result identifies a unique solution $z = z(x,h,\kappa)$  to \eqref{E:clear_function_00} for $(x,h,\kappa)$ fixed.

\begin{lem}\label{L:z_soln}
For each fixed $(x,h,\kappa)$ there is a unique $z = z(x,h,\kappa) > 0$ which solves \eqref{E:clear_function_00}.
\end{lem}

\begin{proof}[Proof of Lemma \ref{L:clear_function_soln}]
By the Inada conditions, $\imutil$ is strictly decreasing with $\imutil(0) = \infty$ and $\imutil(\infty) = 0$. This implies $z\to \omega_U \imutil(z) - (\alpha_I + \alpha_N)\log(z)$ is strictly decreasing in $z$, taking the value $\infty$ at $z=0$ and $-\infty$ at $z=\infty$. The result readily follows.
\end{proof}

Having established existence and uniqueness of a solution  $z = z(x,h,\kappa)$ to \eqref{E:clear_function_00}, we now present the main result of the section, which states existence of a Borel measurable map $h\to \wh{\kappa}(h)$ such that the associated $z$ from \eqref{E:clear_function_00} yields both \eqref{E:U_budget_new} and \eqref{E:clear_real_2}. The proof of Theorem \ref{T:wU0_exists_gen} is given in Appendix \ref{AS:clear_solve}.

\begin{thm}\label{T:wU0_exists_gen}
For each $h$ there exists $\wh{\kappa} = \wh{\kappa}(h)$ such that $z = z(x,h,\wh{\kappa}(h))$ from Lemma \ref{L:z_soln} enforces \eqref{E:U_budget_new} and \eqref{E:clear_real_2}. In particular, the measure $\qprob^h$ of \eqref{E:qprobm_to_qprobh} is well defined. The map $h \to \wh{\kappa}(h)$ can be selected $\mathcal{B}(\reals^d)$ measurable. If the map $w \to w\ \dot{\util}(w)$ is non-decreasing then $\wh{\kappa}(h)$ is unique.
\end{thm}

\begin{rem}
    The condition $w\to w\dot{\util}(w)$ non-decreasing is equivalent to the requirement that the relative risk aversion function $\eta_\util(w) \dfn  -w\ddot{\util}(w)/\dot{\util}(w)$ be bounded above by $1$. This condition is well known to enforce uniqueness of Arrow-Debreu equilibria absent asymmetric information: see \cite{mas1995microeconomic, MR1949437, MR3583456}. For CRRA, it holds provided the relative risk aversion lies in $(0,1]$.
\end{rem}

\subsubsection*{CRRA utility}

Assume CRRA utility with relative risk aversion $\eta_U > 0$. Here, \eqref{E:clear_function_00} specializes to
\begin{equation*}
    \omega_U z^{-1/\eta_U} - (\alpha_I + \alpha_N)\log(z) = \vartheta(x,h) + \kappa,
\end{equation*}
which has ``explicit'' solution
\begin{equation}\label{E:z_crra}
    z = z(x,h,\kappa) = \left(\frac{(\alpha_I +\alpha_N)\eta_U}{\omega_U}\right)^{-\eta_U}  \prodlog{\frac{\omega_U}{(\alpha_I + \alpha_N)\eta_U}e^{\frac{\vartheta(x,h)+\kappa}{\eta_U(\alpha_I+\alpha_N)}} }^{-\eta_U}
\end{equation}
where $\prodlog{\cdot}$ (``Product-Log'' or ``Lambert'' function) is the inverse of $xe^{x}$ on $x\geq  -1$.



\section{Existence of a PCE}\label{S:PCE}

According to Theorem \ref{T:wU0_exists_gen}, there is a Borel mesaurable map $h \to \wh{\kappa}(h)$ such that  if we set (see \eqref{E:ell_T_def} and \eqref{E:gen_cz_new})
\begin{equation}\label{E:gen_cz_nice}
    \cz(x,h) = \frac{z(x,h,\wh{\kappa}(h))\ell(T,x,h)}{\expvs{z(X_T,h,\wh{\kappa}(h))\ell(T,X_T,h)}},
\end{equation}
then $\qprob^h$ in \eqref{E:qprobm_to_qprobh} is well defined and the budget and clearing conditions \eqref{E:U_budget}, \eqref{E:clear} hold. However, this does not yet establish a PCE. In this section we use Theorem \ref{T:wU0_exists_gen} to rigorously prove Theorem \ref{T:main_result}.

\begin{rem}\label{R:signal_realization}
Below we first establish a ``signal-realization'' level PCE. To motivate this step, come back to the optimal investment problems in \eqref{E:opt_U_signal}, \eqref{E:opt_I_signal} and recall the measurability properties for $\pi$ in \eqref{E:AU}, \eqref{E:AI} and $\We^{\pi}$ in Propositions \ref{P:class_equiv}, \ref{P:I_class_equiv}. As $\F^m_0 = \sigma(H)$ and $\F^I_0 = \sigma(H,G_I)$, on an abstract level the agents solve the problems (here we explicitly include the signal dependence in both the strategy and acceptable strategy sets)
\begin{equation*}
    \begin{split}
        (U) &: \inf_{\pi(H)\in\A_U(H)} f_U(H,\pi(H)); \hspace{78pt} f_U(H,\pi(H))\dfn \condexpvs{\util\left((\pi^U_0)'S_0 + \We^{\pi}_{0,T}\right)}{\F^m_0},\\
        (I) &: \inf_{\pi(H,G_I)\in \A_I(H,G_I)} f_I(H,G_I,\pi(H,G_I)); \qquad f_I(H,G_I,\pi(H,G_I)) \dfn  \condexpvs{e^{-\gamma_I \We^{\pi}_{0,T}}}{\F^I_0}.\\
    \end{split}
\end{equation*}

Clearly, the way to solve $\inf_{\pi(X) \in A(X)} f(X,\pi(X))$ for a random variable $X\in\reals^k$ is to first fix $x\in\reals^k$ and solve $\inf_{\pi(x)\in\A(x)} f(x,\pi(x))$, even if $\prob\bra{X=x} = 0$ for all $x\in\reals^d$.  This is what we do, for each agent, finding a market clearing price when establishing the ``signal-realization'' level PCE. Based upon \eqref{E:AU}, \eqref{E:AI}, the construction of $\A_U(h), \A_I(h,g)$ and the measurability properties for $\pi(h),\pi(h,g)$ are clear.  The identification of $f_U(h,\pi(h))$ and $f_I(h,g,\pi(h,g))$ follows from a heuristic argument, but none-the-less is shown the be correct when we use it establish the PCE at the signal level.
\end{rem}

With the above as motivation, we proceed as follows. First, we establish a ``signal realization level'' PCE where all agents have the same information $\filt^B$ but where the signal realizations yield differing endowments/beliefs. Here, we a-priori assume market completeness, and we use classical results on equilibrium in complete markets, in conjunction with the analysis in Section \ref{S:clear_solve}. Second, we will lift up from the signal realization level to the signal level, using the stochastic integration identities given in Appendix \ref{AS:init_enlarge}. Lastly, we will verify our assumption of market completeness.  Throughout, Assumptions \ref{A:X_SDE}, \ref{A:Psi}, \ref{A:U} and \ref{A:init_pos}  are in force.

\subsection*{Signal realization level PCE} Based upon Theorem \ref{T:wU0_exists_gen} the clearing condition \eqref{E:clear_real_2} holds, and using  \eqref{E:qprobm_to_qprobh} the price process of \eqref{E:price_def} is
\begin{equation}\label{E:price_funct_def}
S_t = S(t,X_t,H);\quad S(t,x,h)= \condexpv{\qprob^h}{}{\Psi(X_T)}{X_t = x}.
\end{equation}
Lemma \ref{L:clear_function_soln} below shows $S(t,x,h)$ is well defined, and it is clear  $S$ has volatility (as obtained in the proof of Lemma \ref{L:S_h_prop}) taking the form $\sigma^H_t = \sigma(t,X_t,H)$ for a certain  function $\sigma(t,x,h)$ . For the time being let us assume the following market completeness condition holds.
\begin{tass}\label{TA:complete}
For each $h$, $\sigma(t,x,h)$ is of full rank for Lebesgue almost every $(t,x)\in [0,T]\times \reals^d$.
\end{tass}

Later, we will verify temporary Assumption \ref{TA:complete}, but for now, we take it as given.  Using \eqref{E:exp1} and \eqref{E:H_to_G}, calculation shows for all non-negative $\F^B_T\otimes\B(\reals^{2d})$ measurable maps $\chi^{H,G_I}_T$ that on the set $\cbra{(H,G_I) = (h,g)}$
\begin{equation*}
    \condexpvs{\chi^{H,G_I}_T}{\F^I_0} = \frac{\expvs{\chi^{h,g}_Tp_{C_I}(g-X_T)}}{\expvs{p_{C_I}(g-X_T)}} = \frac{\expvs{\chi^{h,g}_Te^{- \frac{1}{2}X_T'P_I X_T + X_T'P_I g}}}{\expvs{e^{- \frac{1}{2}X_T'P_I X_T + X_T'P_I g}}}.
\end{equation*}
Motivated by this identity, for fixed $(h,g)$, we consider the following optimization problem for $I$, recalling that ``$\xi_T \in \F^B_T$'' means that $\xi_T$ is $\F^B_T$ measurable.
\begin{equation}\label{E:optI_realization_level}
\inf_{\xi_T \in \A_0} \frac{\expvs{e^{-\gamma_I \xi_T - \frac{1}{2}X_T'P_I X_T + X_T'P_I g}}}{\expvs{e^{- \frac{1}{2}X_T'P_I X_T + X_T'P_I g}}};\quad \mathcal{A}_w \dfn \cbra{\xi_T \in \F^B_T \such \xi_T\in L^1(\qprob^h), \expv{\qprob^h}{}{\xi_T} \leq w},
\end{equation}
where $w=0$ in the optimization. Lemma \ref{L:clear_function_soln}  and \eqref{E:gen_cz_nice} imply both $\expvs{\cz(X_T,h)\log(\cz(X_T,h))} < \infty$ and $|X_T|^2 \in L^1(\qprob^h)$, and hence there exists a unique optimizer $\wh{\xi}^I_T$ which satisfies the first order conditions
\begin{equation*}
    -\gamma_I \wh{\xi}^I_T = \frac{1}{2}X_T'P_I X_T - X_T'P_Ig + \log(\cz(X_T,h)) - \expv{\qprob^h}{}{\frac{1}{2}X_T'P_I X_T - X_T'P_Ig + \log(\cz(X_T,h))},
\end{equation*}
yielding $\expv{\qprob^h}{}{\wh{\xi}^I_T} = 0$.  Therefore, by our noise trader convention
\begin{equation*}
    -\gamma_N \wh{\xi}^N_T = \frac{1}{2}X_T'P_I X_T - X_T'P_Ig_N + \log(\cz(X_T,h)) - \expv{\qprob^h}{}{\frac{1}{2}X_T'P_I X_T - X_T'P_Ig_N + \log(\cz(X_T,h))}.
\end{equation*}
where $g_N$ is the realization of $G_N$, which in view of \eqref{E:H_def} is recoverable given realizations of $H,G_I$.  As for $U$, we have for $\F^B_T\otimes\B(\reals^d)$ measurable maps $\chi^H_T$, provided the conditional expectation is well defined, that on $\cbra{H=h}$
\begin{equation*}
    \condexpvs{\chi^H_T}{\F^m_0} = \frac{\expvs{\chi^h_Tp_{C_U}(h-X_T)}}{\expvs{p_{C_U}(h-X_T)}} = \frac{\expvs{\chi^h_Te^{- \frac{1}{2}X_T'P_U X_T + X_T'P_U h}}}{\expvs{e^{- \frac{1}{2}X_T'P_U X_T + X_T'P_U h}}}.
\end{equation*}
This motivates us to consider,  for the initial wealth $w = (\pi^U_0)'\expv{\qprob^h}{}{\Psi(X_T)}$, the optimal investment problem
\begin{equation*}
\sup_{\xi_T \in \A_{w}} \expv{\wh{\prob}^h}{}{\util(\xi_T)};\quad \frac{d\wh{\prob}^h}{d\prob}\big|_{\F^B_T} = \frac{e^{-\frac{1}{2}X_T'P_U X_T + X_T'P_U h}}{\expvs{e^{-\frac{1}{2}X_T'P_U X_T + X_T'P_U h}}}.
\end{equation*}
Here, we follow convention, and set $\expv{\wh{\prob}^h}{}{U(\xi_T)} = -\infty$ if $\expv{\wh{\prob}^h}{}{U(\xi_T)^{-}} = \infty$. To invoke the standard complete market optimization results we need the following lemma, whose proof is given in Appendix \ref{AS:PCE}.

\begin{lem}\label{L:dual_realization_ok}
Define $\dualutil(y) \dfn \sup_{x > 0} (\util(x) - xy)$ for $y>0$.  Then, for all $\lambda> 0$ we have
\begin{equation}\label{E:dual_vf_cond}
    \expv{\wh{\prob}^h}{}{\dualutil\left(\lambda \frac{d\qprob^h}{d\wh{\prob}^h}\big|_{\F^B_T}\right)} < \infty.
\end{equation}
\end{lem}

In view of Lemma \ref{L:dual_realization_ok} we may invoke \cite[Theorems 1,2]{MR2023886} to conclude that an optimal $\wh{\xi}^U_T \in \A(w)$ exists and satisfies
\begin{equation*}
    \wh{\xi}^U_T = \imutil\left(\wh{\lambda}  \frac{d\qprob^h}{d\wh{\prob}^h}\big|_{\F^B_T}\right);\qquad \expv{\qprob^h}{}{\wh{\xi}^U_T} = w.
\end{equation*}
Writing this out in terms of $\prob$ and adjusting $\wh{\lambda}$ we have
\begin{equation*}
    \begin{split}
        \wh{\xi}^U_T &= \imutil\left(\wh{\lambda}z(X_T,h)\right);\qquad \expv{\qprob^h}{}{\imutil\left(\wh{\lambda}z(X_T,h)\right)} = w = (\pi^U_0)'\expv{\qprob^h}{}{\Psi(X_T)}.
    \end{split}
\end{equation*}
Lastly, in view of Theorem \ref{T:wU0_exists_gen} we recall that $\wh{\kappa}(h)$ therein ensures $\wh{\lambda} = 1$. Therefore, we have
\begin{equation*}
    \begin{split}
        \wh{\xi}^U_T &= \imutil\left(z(X_T,h)\right);\qquad \expv{\qprob^h}{}{\imutil\left(z(X_T,h)\right)} = w.
    \end{split}
\end{equation*}
Having identified the optimal terminal wealth for each agent, we invoke temporary Assumption \ref{TA:complete} to obtain the optimal policies.  This gives
\begin{equation*}
    \begin{split}
         -\gamma_I \int_0^T (\wh{\pi}^I(h,g)_u)'dS(u,X_u,h)   &= \frac{1}{2}X_T'P_I X_T - X_T'P_Ig + \log(\cz(X_T,h))\\
         &\qquad \qquad - \expv{\qprob^h}{}{\frac{1}{2}X_T'P_I X_T - X_T'P_Ig + \log(\cz(X_T,h))},\\
         -\gamma_N \int_0^T (\wh{\pi}^N(h,g_N)_u)'dS(u,X_u,h) &= \frac{1}{2}X_T'P_I X_T - X_T'P_Ig_N + \log(\cz(X_T,h))\\
         &\qquad\qquad - \expv{\qprob^h}{}{\frac{1}{2}X_T'P_I X_T - X_T'P_Ig_N + \log(\cz(X_T,h))},\\
         \int_0^T (\wh{\pi}^U(h)_u)'dS(u,X_u,h) &= \imutil\left(z(X_T,h)\right) - \expv{\qprob^h}{}{\imutil\left(z(X_T,h)\right)}.
    \end{split}
\end{equation*}
At this point, we must briefly discuss the optimal policies for $I$, $N$, and how they use their respective ``private'' information.  To this end, note that $X_T\in L^1(\qprob^h)$ and hence there is a (matrix-valued) trading strategy $\wh{\pi}^X(h)$ such that $X_T = \expv{\qprob^h}{}{X_T} + \int_0^T \wh{\pi}^X(h)_u dS(u,X_u,h)$ (see \cite[Chapter 4]{MR2057928}). It is clear by uniqueness\footnote{Strict concavity of the utility functions gives almost sure uniqueness of the optimal terminal wealth, see \cite{MR1722287}. In a diffusive setting, this implies uniqueness of the trading strategy up to $\textrm{Leb}_{[0,T]}\times \prob$ negligible sets (\cite{MR1640352}) which, in general, is as precisely as we can enforce the equilibrium clearing condition in \eqref{E:eq_cond}. However, in our Markovian setting, there is sufficient regularity to deduce uniqueness of the trading strategy functions.}  that
\begin{equation*}
    \wh{\pi}^I(h,g)_t = \frac{1}{\gamma_I}\left(\left(\wh{\pi}^X(h)_t\right)'P_I g - \wh{\theta}(h)_t\right);\quad \wh{\pi}^N(h,g_N)_t = \frac{1}{\gamma_N}\left(\left(\wh{\pi}^X(h)_t\right)'P_I g_N - \wh{\theta}(h)_t\right)
\end{equation*}
where
\begin{equation*}
    \begin{split}
         \int_0^T (\wh{\theta}(h)_u)'dS(u,X_u,h)  &= \frac{1}{2}X_T'P_I X_T  + \log(\cz(X_T,h)) - \expv{\qprob^h}{}{\frac{1}{2}X_T'P_I X_T  + \log(\cz(X_T,h))}.
    \end{split}
\end{equation*}
Now, define the $\mcp(\filt^B)$ measurable strategy
\begin{equation*}
    \begin{split}
        \wh{\pi}(h)_t &\dfn \omega_I \wh{\pi}^I(h,g)_t + \omega_N \wh{\pi}^N(h,g_N)_t + \omega_U \wh{\pi}^U(h)_t - \Pi\\
        &= \wh{\pi}^X(h)'_tP_I\left((\alpha_I +\alpha_N \tau_N)h + \alpha_N \mu_N\right) - (\alpha_I+\alpha_N)\wh{\theta}(h)_t + \omega_U \wh{\pi}^U(h)_t - \Pi.
    \end{split}
\end{equation*}
The associated gains process $\wh{\We}_{0,\cdot}$ is a $\qprob^h$ martingale with terminal value
\begin{equation*}
    \begin{split}
        \wh{\We}_{0,T} &= \left((\alpha_I +\alpha_N \tau_N)h + \alpha_N \mu_N\right)'P_I\left(X_T - \expv{\qprob^h}{}{X_T}\right)\\
        &\qquad - (\alpha_I + \alpha_N)\bigg( \frac{1}{2}X_T'P_I X_T + \log(\cz(X_T,h)) - \expv{\qprob^h}{}{\frac{1}{2}X_T'P_I X_T + \log(\cz(X_T,h))}\bigg)\\
        &\qquad + \omega_U\left(\imutil\left(z(X_T,h)\right) - \expv{\qprob^h}{}{\imutil\left(z(X_T,h)\right)}\right) - \Pi'\left(\Psi(X_T) - \expv{\qprob^h}{}{\Psi(X_T)}\right),\\
        &= -\bigg(\Pi'\Psi(X_T) + \frac{1}{2}(\alpha_I + \alpha_N)X_T'(P_I-P_U)X_T - X_T'\left(\left(\alpha_I+\alpha_N\tau_N\right)P_I - (\alpha_I+
        \alpha_N)P_U\right)h\\
        &\qquad - \alpha_N X_T'P_I\mu_N - \omega_U \imutil(z(X_T,h)) + (\alpha_I + \alpha_N)\log(z(X_T,h))\bigg) + \expv{\qprob^h}{}{(\dots)},\\
        &= 0.
        \end{split}
\end{equation*}
Above, the second equality follows from \eqref{E:gen_cz_nice} and the $(\dots)$ is everything within the parentheses.  The third equality follows from \eqref{E:clear_function_00} as the $\wh{\kappa}(h)$ factors out.   By the $\qprob^h$ martingale property we know $\wh{W}_{0,\cdot} \equiv 0$ and hence $\langle \wh{W}_{0,\cdot}\rangle_{\cdot} \equiv 0$. But this implies that $\textrm{Leb}_{[0,T]}\times \prob$ almost surely we have
\begin{equation*}
    0 = \wh{\pi}(h)_t'\sigma(t,X_t,h)\sigma(t,X_t,h)'\wh{\pi}(h)_t
\end{equation*}
Temporary Assumption \ref{TA:complete} ensures $|\wh{\pi}(h)_t| = 0$ almost surely and hence the realization level PCE holds. To conclude, we have just proved
\begin{prop}
Let Assumptions \ref{A:X_SDE}, \ref{A:Psi}, \ref{A:U}, \ref{A:init_pos} as well as Temporary Assumption \ref{TA:complete} hold. Then there is a signal-realization level PCE.
\end{prop}

\subsection*{Signal level PCE} We now establish the PCE at the signal level. In the previous section, we have constructed  strategies at the realization level which clear the market and by plugging in $H,G_I,G_N$ for $(h,g,g_N)$\footnote{Propositions \ref{P:first_first_result} and \ref{P:first_result} of Appendix \ref{AS:init_enlarge}  shows signal-realization level strategies are appropriately measurable.} it still holds that $\textrm{Leb}_{[0,T]}\times \prob$ almost surely
\begin{equation*}
    \Pi = \omega_I \wh{\pi}^I(H,G_I) + \omega_N \wh{\pi}^N(H,G_N) + \omega_U \wh{\pi}^U(H).
\end{equation*}
However, we must verify the above strategies are optimal at the signal level.  To this end, we start with the insider $I$. To connect $\A_I$ from \eqref{E:AI} with $\A_w$ from \eqref{E:optI_realization_level}, let us write, for a given $\pi\in\A_I$, $\We^{\pi}_{0,T} = \We^{\pi}_{0,T}(H,G_I)$ to stress the dependence on $(H,G_I)$.  We then have the following lemma
\begin{lem}\label{L:I_signal_to_real}
If $\pi \in \A_I$ then there is a set $\E^{\pi}$ of full Lebesgue measure such that $\We^{\pi}_{0,T}(h,g) \in \A_0$ for $(h,g) \in \E^{\pi}$.
\end{lem}

\begin{proof}[Proof of Lemma \ref{L:I_signal_to_real}]
For $\pi\in \A_I$ we know from Proposition \ref{P:I_class_equiv} that $\We^{\pi}_{0,T}(H,G_I)$ is $\F^B_T \otimes \mathcal{B}(\reals^{2d})$ measurable. This implies for Leb a.e. $(h,g)$ that $\We^{\pi}_{0,T}(h,g)$ is $\F^B_T$ measurable, and the integrability condition in $\A_0$ is built into the definition of $\A_I$.
\end{proof}

Next, using Proposition \ref{P:first_result} of Appendix \ref{AS:init_enlarge}, we know $(t,\omega,h,g) \to \wh{\pi}^I(h,g)_t(\omega)$ is $\mcp(\filt^B)\otimes \B(\reals^{2d})$ measurable, and hence $\wh{\pi}^I(H,G_I) \in \A_I$. Therefore, using Proposition  \ref{P:I_class_equiv} we may now prove


\begin{prop}\label{P:optI_same} The strategy
$\wh{\pi}^I(H,G_I)$ is optimal for \eqref{E:opt_I_signal}.
\end{prop}
\begin{proof}[Proof of Proposition \ref{P:optI_same}] As $\wh{\pi}^I(H,G_I) \in \A_I$ it suffices to show it minimizes the conditional expectation in \eqref{E:opt_I_signal}.  To this end, let $\pi\in\A_I$, recall the set $\E^{\pi}$ from Lemma \ref{L:I_signal_to_real}. We then have
\begin{equation*}
    \begin{split}
        \condexpvs{e^{-\gamma_I \We^{\pi}_{0,T}(H,G_I)}}{\F^I_0} &= 1_{\E^{\pi}}(H,G_I)\condexpvs{e^{-\gamma_I \We^{\pi}_{0,T}(H,G_I)}}{\F^I_0}\\
        & = 1_{\E^{\pi}}(H,G_I) \left(\frac{\expvs{e^{-\gamma_I \We^{\pi}_{0,T}(h,g) - \frac{1}{2}X_T'P_I X_T + X_T'P_I g}}}{\expvs{e^{- \frac{1}{2}X_T'P_I X_T + X_T'P_I g}}}\right)\bigg|_{(h,g) = (H,G_I)}\\
        &\geq 1_{\E^{\pi}}(H,G_I) \left(\frac{\expvs{e^{-\gamma_I \wh{\We}^{I}_{0,T}(h,g) - \frac{1}{2}X_T'P_I X_T + X_T'P_I g}}}{\expvs{e^{- \frac{1}{2}X_T'P_I X_T + X_T'P_I g}}}\right)\bigg|_{(h,g) = (H,G_I)}\\
        &= 1_{\E^{\pi}}(H,G_I)\condexpvs{e^{-\gamma_I \wh{\We}^{I}_{0,T}(H,G_I)}}{\F^I_0}\\
        &= \condexpvs{e^{-\gamma_I \wh{\We}^{I}_{0,T}(H,G_I)}}{\F^I_0}.
    \end{split}
\end{equation*}

\end{proof}

Having shown $\wh{\pi}^I(H,G_I)$ is optimal for $I$, it immediately follows by our convention that the noise trader uses the strategy $\wh{\pi}^N(H,G_N) = (\gamma_I/\gamma_N) \wh{\pi}^I(H,G_N)$. Thus, it remains to consider $U$.
Copying Lemma \ref{L:I_signal_to_real}, for each $\pi\in \A_U$ if we write $\We^{\pi}_{T}(H)$ to stress the dependence, we know there is a set $\E^{\pi}$ of full Lebesgue measure such that for $h\in \E^{\pi}$ we have  $\We^{\pi}_T(h) \in \A_{w(h)}$ where $w(h) = (\pi^U_0)'\expv{\qprob^h}{}{\Psi(X_T)}$.    Therefore,

\begin{prop}\label{P:optU_same} The strategy
$\wh{\pi}^U(H)$ solves \eqref{E:opt_U_signal}.
\end{prop}

\begin{proof}[Proof of Proposition \ref{P:optU_same}] As $\wh{\pi}^U(H) \in \A_U$ it suffices to show it maximizes the conditional expectation in \eqref{E:opt_U_signal}.  To this end, let $\pi\in\A_U$. We then have
\begin{equation*}
    \begin{split}
        \condexpvs{\util(\We^\pi_{T}(H))}{\F^m_0} &= 1_{\E^{\pi}}(H)\condexpvs{\util(\We^\pi_{T}(H))}{\F^m_0}\\
        & = 1_{\E^{\pi}}(H) \left(\frac{\expvs{\util(\We^{\pi}_T(h))e^{- \frac{1}{2}X_T'P_U X_T + X_T'P_U h}}}{\expvs{e^{- \frac{1}{2}X_T'P_U X_T + X_T'P_U h}}}\right)\bigg|_{h=H}\\
        &\leq 1_{\E^{\pi}}(H) \left(\frac{\expvs{\util(\wh{\We}^U_T(h)e^{- \frac{1}{2}X_T'P_U X_T + X_T'P_U h}}}{\expvs{e^{- \frac{1}{2}X_T'P_U X_T + X_T'P_U h}}}\right)\bigg|_{h=H}\\
        &= 1_{\E^{\pi}}(H)\condexpvs{\util\left(\wh{\We}^U_T(H)\right)}{\F^m_0}\\
        &= \condexpvs{\util\left(\wh{\We}^U_T(H)\right)}{\F^m_0}.
    \end{split}
\end{equation*}
\end{proof}

Putting all these results together, we have proved
\begin{prop}
Let Assumptions \ref{A:X_SDE}, \ref{A:Psi}, \ref{A:U}, \ref{A:init_pos} as well as Temporary Assumption \ref{TA:complete}. Then there exists a PCE.
\end{prop}

\subsection*{Verifying market completeness} The last thing to do is verify temporary Assumption \ref{TA:complete}, which is done in the following proposition.

\begin{prop}\label{P:on_TA_complete}
Let Assumptions \ref{A:X_SDE}, \ref{A:Psi}, \ref{A:U} and \ref{A:init_pos} hold. Then Temporary Assumption \ref{TA:complete} holds.
\end{prop}

\begin{proof}[Proof of Proposition \ref{P:on_TA_complete}]

The result will follow provided we verify the assumptions in \cite[Theorem 2.3]{MR3131287} for each fixed signal realization $h$. To connect to the notation in \cite{MR3131287}, we have (the symbol on the left side of the equation is the associated one in \cite{MR3131287})
\begin{equation*}
    \begin{split}
        \gamma &\equiv 0, a \equiv 0, f \equiv 0, g\equiv 0, \beta\equiv 0, G(x) = z(x, h, \wh{\kappa}(h))e^{-\frac{1}{2}x'P_u x + x'P_U h}, F(x)  = \Psi(x).
    \end{split}
\end{equation*}
\cite[Assumptions (A1),(A3)]{MR3131287} clearly hold and \cite[Assumption (A2)]{MR3131287} will hold provided we may find $N = N(h)$ such that
\begin{equation}\label{E:KP_bounds_new}
    |F^j_{i}| + |G_{i}| \leq e^{N(1+|x|)};\qquad i = 1,...,d
\end{equation}
where, as in Assumption \ref{A:Psi}, superscript $j$ denotes the $j^{th}$ component of a vector and subscript $i$ is the derivative relative to the $i^{th}$ component of the argument. By Assumption \ref{A:Psi} this clearly holds for $F$. As for $G$, calculations similar to those in Lemma \ref{L:clear_function_soln} show that
\begin{equation*}
    \begin{split}
    &\frac{z_i}{z}(x,h,\wh{\kappa}) = \frac{1}{\alpha_U(\imutil(z)) + \alpha_I + \alpha_N} \bigg( -\sum_{j=1}^d \Pi^j \psi^j_i(x) + (\alpha_I + \alpha_N)((P_I-P_U)x)^i\\
    &\qquad\qquad - (((\alpha_I + \alpha_N\tau_N)P_I - (\alpha_I + \alpha_N)P_U)h)^{i} - \alpha_N (P_I\mu_N)^i\bigg).
    \end{split}
\end{equation*}
As $\alpha_U(\imutil(z))>0$, by Assumption \ref{A:Psi} the magnitude of the right side above is bounded by $C(1+|x| + e^{K|x|})$ for some $C=C(h)$ (which might change from line to line below).  \eqref{E:KP_bounds_new} readily follows from Lemma \ref{L:clear_function_soln} as
\begin{equation*}
    \begin{split}
    |G_i(x)| &\leq z(x,h,\wh{\kappa}(h)) e^{-\frac{1}{2}x'P_Ux + x'P_U h}\left(\left|\frac{z_i}{z}\right|(x,h,\wh{\kappa}) + |(P_U(x-h))^{i}|\right),\\
    &\leq \wt{z}(h,\wh{\kappa}(h)) e^{-\frac{1}{2}x'P_Ux + x'P_U h} C\left(1+|x| + e^{K|x|}\right) \leq C.
    \end{split}
\end{equation*}
Here, the last inequality holds because the map
\begin{equation*}
    x \to e^{-\frac{1}{2}x'P_U x + x'P_U h}\left(1 + |x| + e^{K|x|}\right)
\end{equation*}
is bounded from above on $\reals^d$ by a constant depending only on $h$.
\end{proof}

Putting everything together yields Theorem \ref{T:main_result}, except for the DNREE statement, which is addressed in the next section.


\section{Existence of a DNREE}\label{S:DNREE}

We have shown existence of a PCE under Assumptions \ref{A:X_SDE}, \ref{A:Psi}, \ref{A:U} and \ref{A:init_pos}. We obtain the existence of a DNREE, according to Definition \ref{D:DNREE}, if the signal $H$ is instantaneously recoverable using the information in  $\filt^{B,S}_{+}$, the right-continuous enlargement of $(S,B)$'s natural filtration. In \cite{MR4192554}, we showed that if $U$ has CARA preferences and $\Psi(x) = x$ (or in the one dimensional case, if $\Psi(x)$ is monotonic), then the PCE is a DNREE because the time $t=0$ price is invertible in the signal.  This was a very special outcome of CARA preferences, and while numerical computations suggest the time $0$ price is invertible in $H$, a formal proof of this fact has been elusive.  However, this does not rule out the signal being recoverable at any time $t>0$, and hence existence of a  DNREE in the sense of Definition \ref{D:DNREE}, as we now show in the univariate case.   Results are valid assuming
\begin{ass}\label{A:DNREE_d1}
The dimension is $d=1$.
\end{ass}

Recall the pricing function $S = S(t,x,h)$ from \eqref{E:price_funct_def} (see also \eqref{E:qprobm_to_qprobh}). We first state the following result, whose proof is given in Appendix \ref{AS:DNREE}.
\begin{prop}\label{P:DNREE_1}
Let Assumptions \ref{A:X_SDE},  \ref{A:Psi}, \ref{A:U}, \ref{A:init_pos} and \ref{A:DNREE_d1} hold.  Then the following statement is not true
\begin{enumerate}[]
\item ``There exists $h_1,h_2$ and $\eps > 0$ such that $S(\cdot,h_1) = S(\cdot,h_2)$ on $(0,\eps) \times\reals$.''
\end{enumerate}
\end{prop}

Using this result we finish the proof of Theorem \ref{T:main_result} by proving

\begin{prop}\label{P:PCE_is_DNREE}

Let Assumptions \ref{A:X_SDE},  \ref{A:Psi}, \ref{A:U}, \ref{A:init_pos} and \ref{A:DNREE_d1} hold.   Then the PCE of Theorem \ref{T:main_result} is a DNREE, in the sense of Definition \ref{D:DNREE}.
\end{prop}

\begin{proof}[Proof of Theorem \ref{P:PCE_is_DNREE}]

Let $h_0$ be the (potentially unknown) signal realization and $\eps>0$. By the support theorem \cite{MR0400425} and continuity of $S$, the information in $\F^{B,S}_{+,\eps}$ allows us to see the collection of numbers $\cbra{S(t,x,h_0)}_{t\leq \eps, x\in\reals}$. By Proposition \ref{P:DNREE_1} there is only one $h$ compatible with the collection and hence we can back out $h_0$. As this works for any realization we see that $\F^m_0 \subseteq \F^{B,S}_{+,\eps} \subseteq \F^m_{+,\eps}=\F^m_{\eps}$, for any $\eps > 0$ and hence by right-continuity of both $\filt^m$ and $\filt^{B,S}_{+}$ we see that $\F^m_0 = \F^m_{+,0} = \F^{B,S}_{+,0}$ and the DNREE follows.
\end{proof}

\section{Relative Risk Aversion Asymptotics}\label{S:asympt}

While Theorem \ref{T:main_result} establishes existence of a PCE and DNREE, as $U$'s preferences are general, explicit formulas for equilibrium quantities are difficult to obtain.  This fact remains true even when $U$ has CRRA preferences.   To obtain more concrete expressions, in this section we specify to when $U$ has CRRA utility and consider the risk aversion limits $\eta_U \to \infty$ and $\eta_U \to 0$. Proofs of the propositions below are given in Appendix \ref{AS:asympt}.

We first let $\eta_U\to \infty$.  To gain intuition for the result, consider \eqref{E:clear_function_00} when $\imutil(z) = z^{-1/\eta_U}$.  Formally, if $\eta_U = \infty$ then $z^{-1/\eta_U} = 1$ and by absorbing $\omega_U$ into $\kappa$ we expect that
\begin{equation*}
    z(x,h,\kappa) = e^{-\frac{\vartheta(x,h)+\kappa}{\alpha_I + \alpha_N}}.
\end{equation*}
In view of \eqref{E:gen_cz_nice}, \eqref{E:price_def} and \eqref{E:qprobm_to_qprobh} one expects the limiting price function to be
\begin{equation}\label{E:large_etaU_price}
    S(t,x,h;\infty) = \frac{\condexpvs{\Psi(X_T) e^{-\frac{\vartheta(X_T,h)}{\alpha_I + \alpha_N} - \frac{1}{2}X_T'P_U X_T + X_T'P_U h}}{X_t = x}}{\condexpvs{ e^{-\frac{\vartheta(X_T,h)}{\alpha_I + \alpha_N} - \frac{1}{2}X_T'P_U X_T + X_T'P_U h}}{X_t = x}}.
\end{equation}
This argument in fact leads to the correct price, as the following proposition shows.

\begin{prop}\label{P:large_etaU}
Let Assumptions \ref{A:X_SDE}, \ref{A:Psi} and \ref{A:init_pos} hold, and assume CRRA preferences for $U$ with relative risk aversion $\eta_U$. Then
\begin{equation*}
    \lim_{\eta_U \to \infty} S(t,x,h;\eta_U) = S(t,x,h;\infty),
\end{equation*}
for $S(\cdot,\infty)$ as in \eqref{E:large_etaU_price}
\end{prop}

We next consider $\eta_U \to 0$, obtaining an a priori surprising result. To state it, we first fix $h$ and consider solutions $\kappa = \wh{\kappa}(h;0)$ to
\begin{equation}\label{E:small_etaU_kappa_eqn}
    \frac{1}{\omega_U} = \frac{\expvs{(\pi_0^U)'\Psi(X_T)e^{\frac{(\vartheta(X_T,h)+\kappa)^{-}}{\alpha_I + \alpha_N} - \frac{1}{2}X_T'P_U X_T + X_T'P_U h}}}{\expvs{(\vartheta(X_T,h)+\kappa)^{+}e^{- \frac{1}{2}X_T'P_U X_T + X_T'P_U h}}}.
\end{equation}
As Assumption \ref{A:Psi} implies, for fixed $h$, that $x \to \vartheta(x,h)$ is bounded from below, it is clear there exists a unique $\wh{\kappa}(h;0)$ solving the above.  With this notation we have
\begin{prop}\label{P:small_etaU}
Let Assumptions \ref{A:X_SDE}, \ref{A:Psi} and \ref{A:init_pos} hold, and assume CRRA preferences for $U$ with relative risk aversion $\eta_U$. Then
\begin{equation}\label{E:small_eta_px_limit}
    \lim_{\eta_U \to 0} S(t,x,h;\eta_U) = S(t,x,h;0) \dfn \frac{\condexpvs{\Psi(X_T) e^{\frac{(\vartheta(X_T,h)+\wh{\kappa}(h;0))^{-}}{\alpha_I + \alpha_N} - \frac{1}{2}X_T'P_U X_T + X_T'P_U h}}{X_t = x}}{\condexpvs{ e^{\frac{(\vartheta(X_T,h)+\wh{\kappa}(h;0))^{-}}{\alpha_I + \alpha_N} - \frac{1}{2}X_T'P_U X_T + X_T'P_U h}}{X_t = x}}.
\end{equation}
where $\wh{\kappa}(h;0)$ solves \eqref{E:small_etaU_kappa_eqn}.
\end{prop}

\begin{rem}\label{R:small_etaU} Vanishing risk aversion, $\eta_U \to 0$, intuitively suggests $U$ approaches risk neutrality.  But if $U$ were truly risk neutral, this agent's behavior would be expected to drive the equilibrium price to the risk neutral price (i.e., the risk neutral agent would be expected to dominate in the price formation process) given by
\begin{equation*}
    S_t = \condexpvs{\Psi(X_T)}{\F^m_t} = S(t,X_t,H;rn);\quad S(t,x,h; rn) \dfn \frac{\condexpvs{\Psi(X_T) e^{ - \frac{1}{2}X_T'P_U X_T + X_T'P_U h}}{X_t = x}}{\condexpvs{ e^{- \frac{1}{2}X_T'P_U X_T + X_T'P_U h}}{X_t = x}}.
\end{equation*}
Now, as $\vartheta(x,h)$ is bounded from below in $x$ for $h$ fixed, it is natural to wonder if $\wh{\kappa}(h;0)$ is such that $\vartheta(x,h) + \wh{\kappa}(h;0)\geq 0$.  Indeed, if this were the case then $S(t,x,h;0) = S(t,x,h;rn)$.  However, if this were true, then in view of \eqref{E:small_etaU_kappa_eqn} one would have
\begin{equation*}
    \wh{\kappa}(h;0) = \frac{ \expvs{\left(\omega_U(\pi_0^U)'\Psi(X_T)-\vartheta(X_T,h)\right)e^{- \frac{1}{2}X_T'P_U X_T + X_T'P_U h}} }{\expvs{e^{- \frac{1}{2}X_T'P_U X_T + X_T'P_U h}}}.
\end{equation*}
Therefore, it is natural to define $\wh{\kappa}(h;0)$ by the right side above, and then check if $\vartheta(x,h) + \wh{\kappa}(h;0)\geq 0$.  If so, then the limiting price is the risk neutral price.  However, as can be seen when $X$ is Gaussian, typically, one does not have $\vartheta(x,h) + \wh{\kappa}(h;0)\geq 0$ for all signal realizations $h$, and hence even as the uninformed agent's relative risk aversion vanishes,  the limit price is not the risk neutral price. The explanation for this counter-intuitive phenomenon is that, unlike a truly risk neutral agent, power utility enforces a non-negativity constraint, even in the limiting case.
\end{rem}

\section{Example}\label{S:EX}  We consider a univariate case, d=1, with $X_T = (\mu_X - (1/2)\sigma_X^2)T + \sigma_X B_T$ and $\Psi(x) = e^x$, i.e., the terminal payoff is a geometric Brownian motion and the signal is on the log-payoff, or alternatively the signal is multiplicative, rather than additive, on the terminal payoff. Lastly, we consider $U$ with CRRA utility. Therefore, Assumptions \ref{A:X_SDE}, \ref{A:Psi}, \ref{A:U} and \ref{A:init_pos} are all satisfied.

With all the assumptions met, we obtain a PCE and DNREE from Theorem \ref{T:main_result}.
Next, we examine the impact of coefficients on the price, volatility and market price of risk (MPR), in the context of this example. For numerical computations, we use the baseline set of parameters
\begin{equation}
\label{E:GBM_parameters}
    \begin{split}
        \mu_X &= 0.1,\quad \sigma_X = .3,\quad \Pi = 1,\quad T= 1,\\
        \omega_U &= \omega_I  = \omega_N = \frac{1}{3},\quad \gamma_I = \gamma_N = 3,\quad \eta_U = 5,\\
        C_I &= \sigma_X^2,\quad \tau_N = 1,\quad \mu_N = 0,\quad C_N = C_I.
    \end{split}
\end{equation}
Note in particular that we are assuming the noise trader is a mis-perceiving insider.

\subsection*{Sensitivity with respect to uninformed initial endowment} Figure  \ref{F:Endow} plots the time 0 price, volatility and market price of risk versus the weighted uninformed initial share endowment, for various values of the signal $h$.  Here, we see that price and volatility increase slightly with the endowment, while the MPR decreases significantly. An increase in $U$'s endowment increases his initial wealth. This raises $U$'s demand for the risky asset, leading to an excess aggregate demand for the asset. Adjustment to equilibrium is then consistent with a simultaneous increase in the volatility and decrease in the MPR. The latter effect increases the stock price.

\begin{figure}
\centering
\includegraphics[height=4cm,width=6cm]{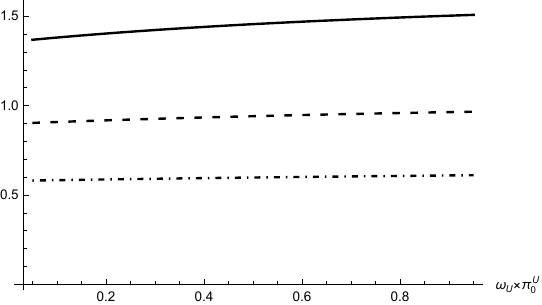}\ \includegraphics[height=4cm,width=6cm]{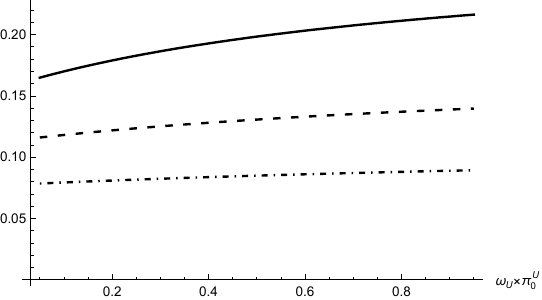}\ \includegraphics[height=4cm,width=8cm]{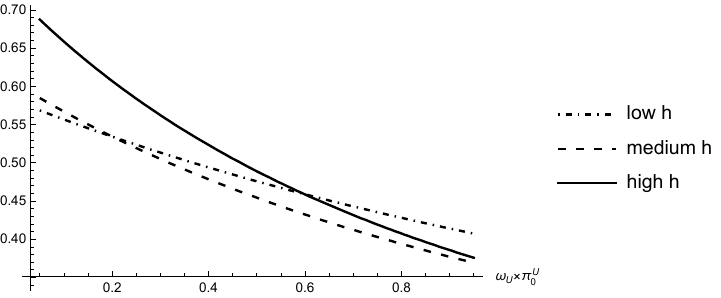}
\caption{Time 0 price (upper left plot), volatility (upper right plot) and market price of risk (lower plot) as a function of the uninformed weighted initial share position for various signal realizations. Parameters are given in \eqref{E:GBM_parameters}.}
\label{F:Endow}
\end{figure}

\subsection*{Sensitivity with respect to uninformed relative risk aversion} Figure \ref{F:riskav} 
plots the time 0 price, volatility and MPR versus the relative risk aversion of $U$ for various values of $h$.  Here, we see a decrease in the price and volatility coefficient as risk aversion increases, and an increase in the MPR. In essence, the more risk averse $U$ requires greater compensation for a given level of risk, leading to an increase in the MPR and a reduction in the asset price. These adjustments are consistent with a simultaneous decrease in the equilibrium volatility. The limit prices in Figure \ref{F:riskav} correspond to those in Propositions \ref{P:large_etaU} and \ref{P:small_etaU}.

\begin{figure}
\centering
\includegraphics[height=4cm,width=6cm]{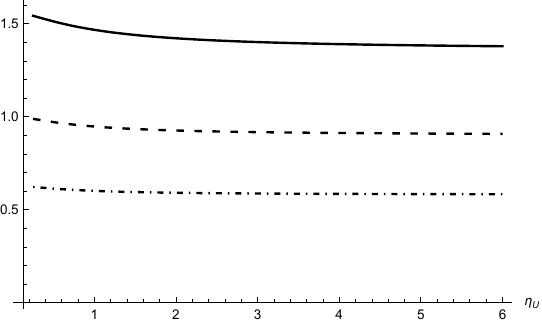}\ \includegraphics[height=4cm,width=6cm]{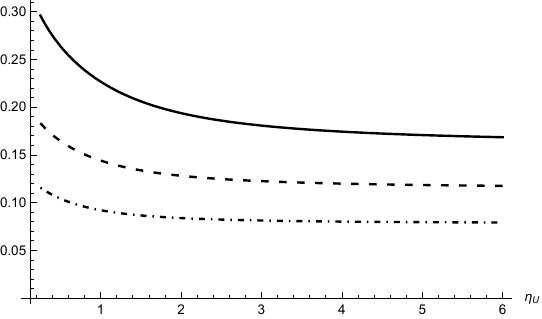}\ \includegraphics[height=4cm,width=8cm]{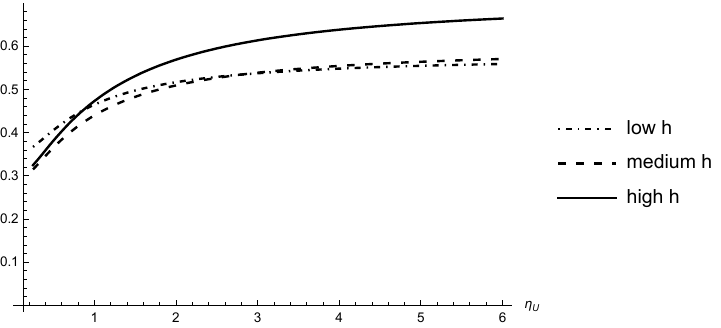}
\caption{Time 0 price (upper left plot), volatility (upper right plot) and market price of risk (lower plot) as a function of the uninformed relative risk aversion for various signal realizations. Parameters are given in \eqref{E:GBM_parameters}.}
\label{F:riskav}
\end{figure}


\subsection*{Sensitivity with respect to noise trader precision} Figure \ref{F:ntprecision}
plots the time 0 price, volatility and MPR versus the noise trader precision $P_N = 1/C_N$, for various values of $h$.\footnote{The computation of $h$ accounts for the fact that the distribution of $h$ depends on $P_N$.} When $P_N\to\infty$ the noise trader becomes an insider, when $P_N \to 0$ the market signal becomes pure noise.  The plots show the reaction depends on the signal realization. For high (low) $h$ the price and volatility decrease (increase) and the MPR increases (decreases). An increase in the noise trader precision implies the market signal is more reliable and the informational advantage of the informed shrinks. The demand for the asset then decreases (increases) if the signal is sufficiently high (low), leading to the effects described. In the limit, as $P_N\rightarrow \infty$, $U$ effectively becomes an insider and the price and MPR both become insensitive to $P_N$. At the opposite extreme, as $P_N\rightarrow 0$, $U$ is unable to extract any useful information from $h$ and information disparity is greatest. Any small change in the neighborhood of $P_N=0$ then has a large impact on heterogeneity and leads to a strong response of equilibrium quantities.


\begin{figure}
\centering
\includegraphics[height=4cm,width=6cm]{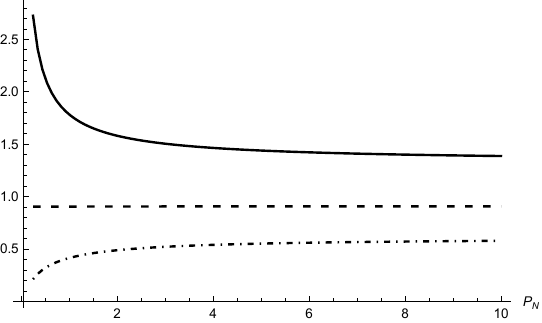}\ \includegraphics[height=4cm,width=6cm]{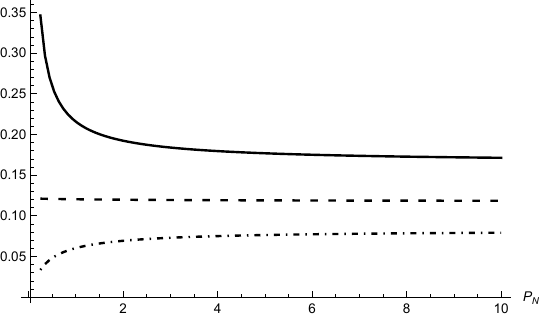}\ \includegraphics[height=4cm,width=8cm]{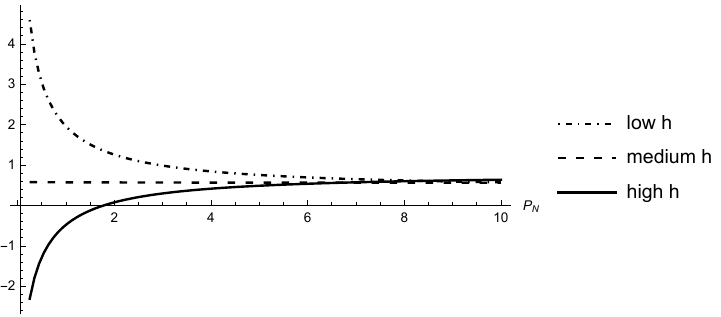}
\caption{Time 0 price (upper left plot), volatility (upper right plot) and market price of risk (lower plot) as a function of the noise trader precision for various signal realizations. Parameters are given in \eqref{E:GBM_parameters}.}
\label{F:ntprecision}
\end{figure}

\appendix

\section{Selected Results on Initial Enlargements}\label{AS:init_enlarge}



In this section, we provide selected results on initially enlarged filtrations. These results are taken almost directly from \cite[Online Appendix]{MR4192554}\footnote{The results in  \cite[Online Appendix]{MR4192554} are themselves a collection of results found in \cite{MR510530,MR1775229, MR3758346,MR2233544, MR3706797, MR3758346}.}  but for ease of reference we repeat them here. We take a complete filtered probability space $\basis$ where $\filt$ satisfies the usual conditions. There is a $d$-dimensional $(\prob,\filt)$ Brownian motion $B$, which has the predictable representation property, but we do not mandate $\filt=\filt^B$. Next, let $\Y \subseteq \reals^m$ be an open set, with Borel sigma-algebra $\B(\Y)$, and for ease of terminology, define

\begin{defn}\label{D:G_prog_mbl} $\theta:[0,T]\times\Omega\times \Y \to \reals^d$ is \emph{$\Y$-predictable} (respectively \emph{$\Y$ -optional}) if $\theta$ is $\mcp(\filt)\otimes \B(\Y)$
(resp. $\OO(\filt)\otimes\B(\Y)$) measurable.
\end{defn}

Let $Y$ be a random variable taking values in $\mathcal{Y}$, and assume

\begin{ass} \label{A:jacod_equiv_abs} For $t\leq T$, $\condprobs{Y\in \cdot}{\F_t}\sim \prob\bra{Y\in\cdot}$ almost surely. Denote by $p^y_t(\omega) = p(t,\omega,y)$ the resultant density,  and by $\lambda$ the unconditional law of $Y$.
\end{ass}

Define $\filtg \dfn \filt \vee \sigma(Y)$. The first lemma contains three results from \cite{MR3758346}.

\begin{lem}\label{L:fontana_facts} Let Assumption \ref{A:jacod_equiv_abs} hold. Then (1) $p$ is $\mathcal{Y}$-optional; (2) $\filtg$ satisfies the usual conditions; (3) $1/p^Y$ is a strictly positive $(\mathbb{P},\mathbb{G})$ martingale with constant expectation $1$; with $\tprob^Y$ denoting the martingale preserving measure defined on $\G_T$ by $d\tprob^Y/d\prob |_{\G_T} = 1/p^Y_T$ we have that $B$ is a $(\tprob^Y,\filtg)$ Brownian motion. In particular, $B$ is a $\filtg$ semi-martingale.
\end{lem}

\begin{proof}[Proof of Lemma \protect\ref{L:fontana_facts}] This follows directly from \cite[Lemma 2.3, Lemma 4.2, Proposition 4.4]{MR3758346} respectively. Part (4) is easily obtained by copying the proofs of \cite[Theorem 3.1, Theorem 3.2]{MR1775229}. 
\end{proof}


The next proposition concerns stochastic integrals when the integrand depends upon $Y$.

\begin{prop}\label{P:first_first_result} Let Assumption \ref{A:jacod_equiv_abs} hold, and let $\theta$ be $\Y$-predictable and such that for $\lambda$ almost every $y\in\Y$ and each $T>0$ we have $\int_0^T |\theta^y_t|^2 dt < \infty$ almost surely.  Then (1) The stochastic integral $\int_0^\cdot (\theta^y_u)'dB_u$ is $\Y$-predictable; (2) The stochastic integral $\int_0^\cdot (\theta^Y_u)'dB_u$ is well defined; and (3) $\int_0^\cdot (\theta^Y_u)'dB_u$ and $\left(\int_0^\cdot (\theta^y_u)'dB_u\right)\big|_{y=Y}$ are indistinguishable.
\end{prop}

\begin{proof}[Proof of Proposition \protect\ref{P:first_result}]  $(1)$ is proved in \cite[Proposition 5]{MR510530} when $\expvs{(\int_0^T |\theta(t,\cdot,y)|^2 dt)^{1/2}} < \infty$ (and noting the integral sample paths are continuous). For the general case, set $\theta_n = \theta 1_{|\theta|\leq n}$ and write $M^n$ as the resultant $\Y$-predictable map. Clearly, we have $\prob \text{-} \lim_{n,m
\rightarrow\infty} \int_0^T |\theta_n(t,\cdot,y)-\theta_m(t,\cdot,y)|^2 dt = 0$, and hence by \cite[Prop. 3.2.26]{MR1121940}, $\mathbb{P} \text{-} \lim_{n,m\rightarrow\infty} \sup_{t\leq T}\left|M_n(t,\cdot,y)-M_m(t,\cdot,y)\right| = 0$. The result follows using \cite[Proposition 1]{MR510530} with $\F,\prob$ there-in being $\mcp(\filt)$ and $\prob\times \text{Leb}_{[0,T]}$ respectively. For part $(2)$, we first note that by \cite[Lemma 4.2]{MR2233544}, $\theta^Y\in\mcp(\filtg)$. Next, Lemma \ref{L:fontana_facts} implies $B$ is a $(\tprob^Y,\filtg)$ Brownian motion, and hence the result will follow if $\prob\bra{\int_0^T|\theta^Y_u|^2 du <\infty} = 1$ as $\tprob^Y\sim \prob$. By Fubini we know that $1_{\int_0^T |\theta^Y_u|^2 du <\infty} = (1_{\int_0^T |\theta^y_u|^2 du <\infty})|_{y=Y}$, and that $(\omega,g) \to 1_{(\int_0^T |\theta^y_u|^2 du)(\omega) <\infty}$ is $\F_T\otimes\B(\Y)$ measurable. Thus, from \cite[Equation (4.1)]{MR3758346} we conclude
\begin{equation}  \label{E:first_1}
\prob\bra{\int_0^T |\theta^Y_u|^2 du <\infty} = \int_{\Y}\expvs{p^y_T 1_{\int_0^T |\theta^y_u|^2 du <\infty}}\lambda(dy)= \int_{\Y}\expvs{p^y_T}\lambda(dy) = 1,
\end{equation}
where the last equality follows from \cite[Equation (4.1)]{MR3758346} applied to $f\equiv 1$.

That part (3) holds is stated in the proof of \cite[Proposition 4.10]{MR3758346} as following from a) an application of the monotone convergence theorem and b) \cite[Proposition 5]{MR510530} combined with the dominated convergence theorem for stochastic integrals (c.f. \cite[IV.Theorem 32]{MR2273672}). Part (3) is also implicitly used in the proof of \cite[Corollary 2.10]{MR1632213}. However, for the sake of clarity, we will offer a detailed sketch.

First, assume $\theta (t,\omega ,y)=\psi (t,\omega )h(y)$ where $\psi \in\mcp(\filt)$, $h\in\B(\Y)$ are bounded. Considering integration with respect to the $(\prob,\filt)$ -Brownian motion $B$, it follows that $(\int_0^\cdot(\theta^y_u)'dB_u)|_{y=Y} = h(Y)\int_{0}^{\cdot}\psi_u'dB_{u}$. Next, considering integration with respect to the $(\wt{\prob}^{Y}, \filtg)$-Brownian motion $B$ we have $\int_{0}^{t}(\theta^Y_u)'dB_{u} = h(Y)\int_{0}^{\cdot}\psi_u'dB_{u}$. The result follows by path-continuity. Next, let bounded $\left\{\theta _{n}\right\}$ converge (bounded-ly) to a bounded $\theta $. Write the associated integrals as $M_{n},M$ and recall by \cite[Lemma 4.2]{MR3758346} that $M_n,M$ are $\OO(\filt)\otimes \B(\Y)$ measurable. For each $t\leq T$
\begin{equation*}
\begin{split}
& \expv{\wt{\prob}^{Y}}{}{\left( M(t,\cdot,Y)-\int_{0}^{t}\theta (u,\cdot ,Y)'dB_{u}\right) ^{2}} \\ & \qquad \leq 2\expv{\wt{\prob}^{Y}}{}{\left( M(t,\cdot ,Y)-M_{n}(t,\cdot ,Y)\right) ^{2}} +2\expv{\wt{\prob}^{Y}}{}{\left( \int_{0}^{t}(\theta _{n}(u,\cdot,Y)-\theta (u,\cdot ,Y))'dB_{u}\right) ^{2}} .
\end{split}
\end{equation*}
First,
\begin{equation*}
\begin{split}
\expv{\tprob^{Y}}{}{\left( M(t,\cdot,Y)-M_{n}(t,\cdot ,Y)\right) ^{2}} & =\expvs{\frac{1}{p_{t}^{Y}}\left( M(t,\cdot ,Y)-M_{n}(t,\cdot ,Y)\right) ^{2}} \\
& =\expvs{\int_{\Y}\left( M(t,\cdot,y)-M_{n}(t,\cdot ,y)\right) ^{2}\lambda(dy)}\\
& =\expvs{\int_{\Y}\left(\int_{0}^{t}\left\vert \theta (u,\cdot ,y)-\theta _{n}(u,\cdot,y)\right\vert ^{2}du\right) \lambda(dy)} \\
& =\expv{\tprob^Y}{}{\int_{0}^{t}\left\vert\theta (u,\cdot ,Y)-\theta _{n}(u,\cdot ,Y)\right\vert ^{2}du}.
\end{split}
\end{equation*}
Above we have used the definition of $p^{Y}$ and the It\^{o} isometry. Similarly
\begin{equation*}
\expv{\tprob^Y}{}{\left(\int_{0}^{t}(\theta _{n}(u,\cdot ,Y)-\theta (u,\cdot ,Y))'dB_{u}\right) ^{2}} =\expv{\tprob^Y}{}{\int_{0}^{t}\left\vert \theta (u,\cdot ,Y)-\theta _{n}(u,\cdot
,Y)\right\vert ^{2}du}.
\end{equation*}
The bounded convergence theorem implies almost surely for $t\leq T$ that $M(t,\cdot ,Y)-\int_{0}^{t}\theta (u,\cdot ,Y)'dB_{u}=0$. As $\theta$ is bounded, $\int_{0}^{\cdot }\theta (u,\cdot ,Y)'dB_{u}$ is a $(\tprob^Y,\filtg)$ martingale. But, this implies $M(t,\cdot ,Y)$ is also a $(\tprob^Y,\filtg)$ martingale. As martingale representation holds with respect to $B$, we
deduce $M(\cdot ,\cdot ,Y)$ has continuous paths and hence $M(\cdot,\cdot ,Y)$ and $\int_{0}^{\cdot }\theta (u,\cdot ,Y)'dB_{u}$ are indistinguishable. The monotone class theorem gives the result for bounded $\theta $. We now extend to $\theta $ such that $\int_{0}^{T}|\theta (u,\cdot,Y)|^{2}du<\infty $. For each $t,n$
\begin{equation}\label{E:M_comparison_a}
\begin{split}
&M(t,\cdot ,Y)-\int_{0}^{t}\theta (u,\cdot ,Y)' dB_{u}=\left(\int_{0}^{t}(\theta (u,\cdot ,y)1_{|\theta (u,\cdot ,y)|\geq n})'dB_{u}\right) \big|_{y=Y}\\
&\qquad\qquad -\int_{0}^{t}(\theta (u,\cdot ,Y)1_{|\theta(u,\cdot ,Y)|\geq n})^{\prime }dB_{u}.
\end{split}
\end{equation}
We first handle the right-most term above. By construction of $p^{Y}$, for each $\eps >0$
\begin{equation*}
\tprob^Y\bra{\int_{0}^{T}|\theta (u,\cdot,Y)|^{2}1_{|\theta (u,\cdot ,Y)|\geq n}du\geq \eps} =\int_{\Y}\expvs{1_{\int_{0}^{T}|\theta (u,\cdot
,y)|^{2}1_{|\theta (u,\cdot ,y)|\geq n}du\geq \eps}}\lambda(dy) .
\end{equation*}
Since for $\lambda$ a.e. $y\in \Y$, $\int_{0}^{T}|\theta(u,\cdot,y)|^{2}1_{|\theta (u,\cdot ,y)|\geq n}du\rightarrow 0$ almost surely as $n\uparrow \infty $, two applications of the dominated convergence theorem allow us to conclude that $\lim_{n\uparrow \infty }\int_{0}^{T}|\theta (u,\cdot ,Y)|^{2}1_{|\theta (u,\cdot ,Y)|\geq n}du=0$ in $\tprob^Y$ probability. Therefore, by \cite[Prop. 3.2.26]{MR1121940} we know that in $\tprob^Y$ probability
\begin{equation*}
\lim_{n\uparrow \infty }\sup_{t\leq T}\left\vert \int_{0}^{t}(\theta(u,\cdot ,Y)1_{|\theta (u,\cdot ,Y)|\geq n})^{\prime }dB_{u}\right\vert =0.
\end{equation*}
As for the first term on the right side of \eqref{E:M_comparison_a} set $M_{n}(t,\cdot ,y)\dfn \int_{0}^{t}(\theta (u,\cdot ,y)1_{|\theta (u,\cdot,y)|\geq n})'dB_{u}$. Since for each $y\in \Y$, $\int_{0}^{T}|\theta (u,\cdot ,y)|^{2}1_{|\theta (u,\cdot ,y)|\geq n}du$ converges to $0$ almost surely, we again deduce from \cite[Prop 3.2.26]{MR1121940} that in $\prob$ probability $\sup_{t\leq T}|M_{n}(t,\cdot,y)|\rightarrow 0$. As $M_{n}$ is $\Y$-optional,
\begin{equation*}
\tprob^Y\bra{\sup_{t\leq T}|M_{n}(t,\cdot ,Y)|\geq \eps} =\int_{\Y}\expvs{1_{\sup_{t\leq T}|M_{n}(t,\cdot ,y)|\geq \eps}}\lambda(dy) ,
\end{equation*}%
so that $\sup_{t\leq T}|M_{n}(t,\cdot ,Y)|\rightarrow 0$ in $\tprob^Y$ probability. Thus, by taking subsequences where the convergence takes place almost surely $\tprob^Y$ and hence $\prob$, we deduce from \eqref{E:M_comparison_a} that $\sup_{t\leq T}\vert M(t,\cdot ,Y)-\int_{0}^{t}\theta (u,\cdot ,Y)'dB_{u}\vert =0$ almost surely, finishing the result.

\end{proof}

The next result concerns martingale representation. To state it, assume

\begin{ass} \label{A:phi} $\phi = \phi(\omega,y)$ is a $\F_T\otimes\B(\Y)$ measurable function such that $\expvs{|\phi(\cdot,y)|} < \infty$ for each $y\in\Y$.
\end{ass}

Denote by $\theta = \theta^y$ the process $\theta^y\in\mcp(\filt)$ for each $y\in\Y$ such that \begin{equation} \label{E:phi_theta}
M^y_\cdot \dfn \condexpvs{\phi(\cdot,y)}{\F_t} = M^y_0 + \int_0^\cdot (\theta^y_u)'dB_u.
\end{equation}

We then have the following intuitive result.

\begin{prop}\label{P:first_result} Let Assumptions \ref{A:jacod_equiv_abs} and \ref{A:phi} hold, and let $\theta$ be from \eqref{E:phi_theta}. Then $\theta$ is $\Y$-predictable and hence satsifies the assumptions of Proposition \ref{P:first_first_result}.
\end{prop}

\begin{cor} \label{C:first_result} If additionally $\theta$ is strictly positive almost surely then the same conclusions hold for $\nu = \nu^y$ defined by
\begin{equation*}
\frac{\phi(\cdot,y)}{\expvs{\phi(\cdot,y)}} = \mathcal{E}\left(\int_0^\cdot (\nu^y_u)'dB_u\right)_T.
\end{equation*}
\end{cor}

\begin{proof}[Proof of Proposition \protect\ref{P:first_result}] From \cite[Proposition 3]{MR510530}, we can take $M = M^y$ in \eqref{E:phi_theta} to be a RCLL and $\B(\Y)$ measurable version of the $\filt$-optional projection of $\phi(\cdot,y)$ (see also the proof of \cite[Lemma 4.2]{MR3758346}). The result then follows from \cite[Proposition A.1]{MR3758346}.
\end{proof}

\begin{proof}[Proof of Corollary \protect\ref{C:first_result}] It is clear that $\nu = \theta / M$. Thus, by the results on $\theta$ above (in particular the connection the proof of $(4)$ which proved indistinguishability for general $\theta$), it suffices to prove that $\int_0^T |\nu^Y_u|^2 du < \infty$ almost surely. But, this will follow provided $\inf_{t\leq T} M^{Y}_t > 0$ almost surely. But, this latter fact follows using the same calculations as in \eqref{E:first_1}, but now for the random variable $1_{\inf_{t\leq T} M^y_t > 0}$, which is almost surely $1$ for all $y$ since $M^y$ has continuous paths.
\end{proof}

Lastly, we relate conditional expectations between $\filt$ and $\filtg$ not necessarily under the measure $\prob$.  To state the result, fix $T>0$ and let $Z=Z^y$ be a strictly positive, $\F_T\otimes \B(\Y)$ measurable map such that for $\lambda$ a.e.  $y\in\Y$ we have $\expvs{Z^y_T} = 1$, and define the measure $\qprob^y$ on $\F_T$ by $d\qprob^y/d\prob |_{\F_T} = Z^y_T$.  Define the measure $\qprob^{\filtg}$ on $\G_T$ through $d\qprob^{\filtg}/d\prob |_{\G_T} = Z^Y_T/p^Y_T$ and note that $\qprob^{\filtg}$ is well defined because
\begin{equation*}
    \expvs{\frac{Z^Y_T}{p^Y_T}} = \expvs{\condexpvs{\frac{Z^Y_T}{p^Y_T}}{\F_T}} = \int_{\Y} \expvs{Z^y_T}\lambda(dy) = 1.
\end{equation*}
Next, we pause to briefly discuss how conditional expectations may be defined for non-integrable random variables. Let $\probtriple$ be a generic complete probability space, let  $\Hcal \subseteq \F$ be a sigma-algebra, and recall that conditional expectations of non-negative random variables with respect to $\Hcal$ are always well defined by the monotone convergence theorem.  To extend beyond non-negative random variables, but also without requiring the $L^1(\prob)$ property,  let $X$ be a random variable and set
\begin{equation*}
    \Delta_{X}(\F,\Hcal;\prob) \dfn \cbra{\condexpv{\prob}{}{X^+}{\Hcal} = \infty} \bigcup \cbra{\condexpv{\prob}{}{X^-}{\Hcal} = \infty},
\end{equation*}
and the class of functions
\begin{equation*}
    \begin{split}
        L^1(\F,\Hcal;\prob) &\dfn \cbra{X\in L^0(\F;\prob) \such \prob\bra{\Delta_{X}(\F,\Hcal;\prob)}= 0}.
    \end{split}
\end{equation*}
For $X\in L^1(\F,\Hcal;\prob)$ the conditional expectation of $X$ given $\Hcal$ is well defined setting $\condexpvs{X}{\Hcal} \dfn \condexpv{\prob}{}{X^+}{\Hcal} - \condexpv{\prob}{}{X^-}{\Hcal}$.   With this notation we now state

\begin{lem}\label{L:market_filt_cond_exp} Fix  $0\leq s < t\leq T$, and let $\chi^y_t$ be a $\F_t\otimes \B(\Y)$ measurable map  such that for $\lambda$ a.e.  $y\in\Y$ we have $\expv{\qprob^y}{}{|\chi^y_t|}<\infty$. Then $\condexpv{\qprob^{\filtg}}{}{\chi^Y_t}{\G_s}$ is well defined, and
\begin{equation*}
\condexpv{\qprob^{\filtg}}{}{\chi^Y_t}{\G_s} = \left(\condexpv{\qprob^y}{}{\chi^y_t}{\F_s}\right)\big|_{y=Y}.
\end{equation*}
\end{lem}

\begin{proof}[Proof of Lemma \protect\ref{L:market_filt_cond_exp}]  First, note by Tonelli's theorem that the map
\begin{equation*}
    (\omega,y) \to \condexpv{\qprob^y}{}{\chi^y_t}{\F_s} = \frac{\condexpvs{Z^y_T \chi^y_t}{\F_s}}{\condexpvs{Z^y_T}{\F_s}},
\end{equation*}
is $\F_s\otimes \B(Y)$ measurable.  Now, assume $\chi^Y_t$ is bounded. Let $A_s \in \F_s$, $\psi$ a smooth bounded function on $\Y$, and recall that $\lambda$ is the law of $Y$. We have
\begin{equation*}
\begin{split}
\expv{\qprob^{\filtg}}{}{\chi^Y_t 1_{A_s}\psi(Y)} &= \expvs{\frac{Z^Y_T}{p^Y_T}\chi^Y_t 1_{A_s} \psi(Y)} = \expvs{1_{A_s} \condexpvs{\frac{Z^Y_T}{p^Y_T} \chi^Y_t \psi(Y)}{\F_T}}\\
&= \int_{\Y}\psi(y)\expvs{\chi^y_t 1_{A_s} Z^y_T}\lambda(dy) = \int_{\Y}\psi(y) \expv{\qprob^y}{}{\chi^y_t 1_{A_s}}\lambda(dy)\\
&= \int_{\Y}\psi(y)\expv{\qprob^y}{}{1_{A_s}\condexpv{\qprob^y}{}{\chi^y_t}{\F_s}}\lambda(dy).\\
\expv{\qprob^{\filtg}}{}{\left(\condexpv{\qprob^y}{}{\chi^y_t}{\F_s}\right)\big|_{y=Y} 1_{A_s}\psi(Y)} &= \expvs{\frac{Z^Y_T}{p^Y_T}\left(\condexpv{\qprob^y}{}{\chi^y_t}{\F_s}\right)\big|_{y=Y} 1_{A_s} \psi(Y)}\\
&= \int_{\Y}\psi(y)\expvs{\condexpv{\qprob^y}{}{\chi^y_t}{\F_s} 1_{A_s} Z^y_T}\lambda(dy)\\
&= \int_{\Y}\psi(y)\expv{\qprob^y}{}{1_{A_s}\condexpv{\qprob^y}{}{\chi^y_t}{\F_s}}\lambda(dy).
\end{split}
\end{equation*}
This gives the result for bounded $\chi^Y_t$.  By monotone convergence we have
\begin{equation*}
    \condexpv{\qprob^{\filtg}}{}{\left(\chi^Y_t\right)^{\pm}}{\G_s}  = \left(\condexpv{\qprob^y}{}{\left(\chi^y_t\right)^{\pm}}{\F_s}\right)\big|_{y=Y}.
\end{equation*}
$\chi^y_t\in L^1(\qprob^y)$ for all $y$ implies that $\Delta_{\chi^Y_t}(\G_t,\G_s;\qprob^{\qprob}) = \emptyset$ and hence $\chi^Y_T \in L^1(\G_t,\G_s,\qprob^{\filtg})$ and
\begin{equation*}
    \begin{split}
        \condexpv{\qprob^{\filtg}}{}{\chi^Y_t}{\G_s}  &= \left(\condexpv{\qprob^y}{}{\left(\chi^y_t\right)^{+}}{\F_s} - \condexpv{\qprob^y}{}{\left(\chi^y_t\right)^{-}}{\F_s}\right)\big|_{y=Y},\\
        &= \left(\condexpv{\qprob^y}{}{\chi^y_t}{\F_s}\right)\big|_{y=Y},
    \end{split}
\end{equation*}
which gives the result.
\end{proof}


\section{Proofs from Section \ref{S:gen_U_arb_w}}\label{AS:gen_U_arb_w}

\begin{proof}[Proof of Lemma \ref{L:S_h_prop}] Using Corollary \ref{C:first_result} in Appendix \ref{AS:init_enlarge}, we deduce $d\qprob^h/d\prob |_{\F^B_T} = \E\left(\int_0^\cdot (\nu^h_u)'dB_u\right)_T$ where $(t,\omega,h) \to \nu^h_t(\omega)$ is $\mcp(\filt^B)\otimes \B(\reals^d)$ measurable. Additionally, by Girsnov's theorem, $B^{\qprob^h}_t \dfn B_t - \int_0^t \nu^h_u du$ is a $(\qprob^h,\filt^B)$ Brownian motion.  Now, for $i=1,...,d$ we may write
\begin{equation*}
    (S^h_t)^{(i)} = \frac{M(i)^h_t}{N^h_t};\qquad M(i)^h_t = \condexpvs{(\Psi(X_T))^{(i)} \cz^h_T}{\F^B_t};\quad N^h_t = \E\left(\int_0^\cdot (\nu^h_u)'dB_u\right)_t.
\end{equation*}
Proposition \ref{P:first_result} in Appendix \ref{AS:init_enlarge} allows us to write $M(i)^h_t = \int_0^t (\theta(i)^h_u)' dB_u$ for some $\mcp(\filt^B)\otimes \B(\reals^d)$ measurable map $\theta(i)^h_t(\omega)$, and also shows that both $M(i)^h_t(\omega)$ and $N^h_t(\omega)$ are $\mcp(\filt^B)\otimes \B(\reals^d)$ measurable maps. Given this, Ito's formula shows
\begin{equation*}
    d\left(S^{h}_t\right)^{(i)} = \frac{1}{N^h_t}\left(\theta(i)^h_t - M(i)^h_t \nu^h_t\right)' dB^{\qprob^h}_t,
\end{equation*}
and hence $dS^h_t = \sigma^h_t dB^{\qprob^h}_t$ for the $d\times d$ matrix valued map
\begin{equation*}
    \left(\sigma^h_t\right)^{(ij)} = \frac{1}{N^h_t}\left(\theta(i)^h_t - M(i)^h_t\nu^h_t\right)^{(j)},\quad i,j = 1,...,d,
\end{equation*}
and that clearly $(t,\omega,h)\to \sigma^h_t(\omega)$ is $\mcp(\filt^B)\otimes \B(\reals^d)$ measurable.  This gives the result as $\Sigma^h_t = \sigma^h_t (\sigma^h_t)'$.

\end{proof}


\begin{proof}[Proof of Proposition \ref{P:class_equiv}]
Throughout we use Lemma \ref{L:S_h_prop}, and recall from it's proof the $(\filt^B,\qprob^h)$ Brownian motion $B^{\qprob^h}$ defined by
\begin{equation*}
    dB^{\qprob^h}_t = dB_t - \nu^h_t dt;\qquad \frac{d\qprob^h}{d\prob}\big|_{\F^B_T} = \cz^h_T = \E\left(\int_0^\cdot (\nu^h_u)'dB_u\right)_T,
\end{equation*}
and the $\mcp(\filt^B)\otimes \B(\reals^d)$ measurable map $\sigma^h_t(\omega)$ enforcing $dS^h_t = \sigma^h_t dB^{\qprob^h}_t = \sigma^h_t \left(dB_t - \nu^h_t dt\right)$. Proposition \ref{P:first_first_result} implies  $(t,\omega,h) \to S^h_t(\omega)$ is $\mcp(\filt^B)\otimes \B(\reals^d)$ measurable, and hence
\begin{equation*}
    S^H_t = \left(S^h_0 + \int_0^t \sigma^h_u dB^{\qprob^h}_u\right)\big|_{h=H}.
\end{equation*}
Using \cite[Theorem 3.2]{MR1775229} we know that $B$ is a $(\tprob^H,\filt^m)$ Brownian motion, and we claim that $B^{\qprob^m}$ defined by $B^{\qprob^m}_t \dfn B_t - \int_0^t \nu^H_u du$ is a $(\qprob^m,\filt^m)$ Brownian motion. Indeed, this follows using \eqref{E:qprobm_to_qprobh} which shows (here, we are implicilty assuming we have localized to ensure the appropriate integrability)
\begin{equation*}
    \begin{split}
        \condexpv{\qprob^m}{}{B^{\qprob^m}_t}{\F^m_s} &= \condexpv{\qprob^m}{}{B_t - \int_0^t \nu^H_u du}{\F^m_s},\\
        &= -\int_0^s \nu^H_u du + \left(\condexpv{\qprob^h}{}{B_t - \int_s^t \nu^h_u du}{\F^B_s}\right)\big|_{h=H},\\
        &= -\int_0^s \nu^H_u du + \left(\condexpv{\qprob^h}{}{B^{\qprob^h}_t + \int_0^s \nu^h_u du}{\F^B_s}\right)\big|_{h=H},\\
        &= \left(B^{\qprob^h}_s\right)\big|_{h=H} = B_s - \int_0^s \nu^H_u du = B^{\qprob^m}_s.
    \end{split}
\end{equation*}
By path continuity it is clear that the $\filt^m$ quadratic variation of $B^{\qprob^m}$ at any time $t$ is $t$ and hence $B^{\qprob^m}$ is a $(\qprob^m,\filt^m)$ Brownian motion, and under the given assumption on $\Sigma^h$, Proposition \ref{P:first_first_result} ensures
\begin{equation*}
    S^H_t = \left(S^h_0 + \int_0^t \sigma^h_u\left(dB_u - \nu^h_u du\right)\right)\big|_{h=H} = S^H_0 + \int_0^t \sigma^H_u\left(dB_u - \nu^H_u du\right) = S^H_0 + \int_0^t \sigma^H_u dB^{\qprob^m}_u.
\end{equation*}
This in turn shows that $S^H$ has $\filt^m$ quadratic variation $\langle S^H, S^H \rangle_t = \int_0^t \Sigma^H_u du$. Given this, we now prove the statements of the Lemma, starting with (i). By \cite[Lemma 4.2]{MR2233544} we know that $\pi\in \A_U$ is (a) $\mcp(\filt^m)$ measurable, and (b) such that $\int_0^T (\pi^H_t)'\Sigma^H_t \pi^H_t dt < \infty$ a.s., hence $S^H$ integrable. This shows $\A_U \subseteq \wh{\A}_U$.  Next, \cite[Proposition 4.22]{Aksamit2017}\footnote{\cite[Proposition 4.22]{Aksamit2017} is stated in dimension $1$ but can easily be updated to general dimensions.} yields, for each $\mcp(\filt^m)$ predictable $\wh{\pi}\in\wh{\A}_U$, a $\mcp(\filt^B)\otimes\B(\reals^d)$ measurable modification $\pi$.  Therefore, if $\wh{\pi}$ is also $S^H$ integrable then $0 = \int_0^T (\pi_t-\wh{\pi}_t)'\Sigma^H_t(\pi_t-\wh{\pi}_t)dt$ implies $\pi$ is $S^H$ integrable and that the resultant wealth processes are indistinguishable. This proves (i).

As for $(ii)$, let $\pi\in\A_u$. Given $\pi\in\wh{\A}_U$, \cite[Lemma 4.2]{MR3758346} implies $\We^{\pi}_0,\cdot$ is $\OO(\filt^B)\otimes \B(\reals^d)$ measurable, and hence $\mcp(\filt^B)\otimes \B(\reals^d)$ mesaurable by path continuity, and part (1) of Proposition \ref{P:first_first_result} shows
\begin{equation*}
    \left(\int_0^\cdot (\pi^h_u)'dS^h_u\right)\big|_{h=H} = \left(\int_0^{\cdot} (\pi^h_u)'\sigma^h_u\left(dB_u - \nu^h_u du\right)\right)_{h=H},
\end{equation*}
is also $\mcp(\filt^B)\otimes \B(\reals^d)$ measurable. Furthermore, as
\begin{equation*}
    \We^{\pi}_{0,\cdot} = \int_0^{\cdot} (\pi^H_u)'\sigma^H_u\left(dB_u - \nu^H_u du\right),
\end{equation*}
the indistinguishability of $\We^{\pi}_{0,\cdot}$ and $\left(\int_0^\cdot (\pi^h_u)'dS^h_u\right)\big|_{h=H}$ follows by part (3) of Proposition \ref{P:first_first_result}. Lastly, the conditional expectation statement in (iii) follows from Lemma \ref{L:market_filt_cond_exp} at $s=0, t=T$ and with $\chi^h_T = \int_0^T (\pi^h_u)'dS^h_u$ so that $\chi^H_T = \We^{\pi}_{0,T}$.

\end{proof}


\section{Proofs from Section \ref{S:clear_solve}}\label{AS:clear_solve}

Throughout, Assumptions \ref{A:X_SDE}, \ref{A:Psi}, \ref{A:U} and \ref{A:init_pos} hold. Next, recall from Lemma \ref{L:z_soln} that for fixed $(x,h,\kappa)$ there is a unique $z = z(x,h,\kappa)$ solving \eqref{E:clear_function_00}.  Before proving Theorem \ref{T:wU0_exists_gen} we must establish two lemmas, but we begin with a remark.

\begin{rem}\label{R:alt_forms}
Using \eqref{E:ell_T_def}, \eqref{E:chi_def} and \eqref{E:vartheta_def}, we obtain the alternate forms of \eqref{E:clear_function_00}
\begin{equation}\label{E:clear_function}
\begin{split}
\omega_U \imutil(z) - (\alpha_I + \alpha_N)\log(z) &= \chi(x,h) + (\alpha_I + \alpha_N)\log(\ell(T,x,h)) + \kappa_1,\\
\omega_U \imutil(z) - (\alpha_I + \alpha_N)\log(z) &= (\alpha_I +\alpha_N)\phi(x,h) + \kappa_2,
\end{split}
\end{equation}
where $\kappa_1,\kappa_2$ are suitably adjusted constants, and
\begin{equation}\label{E:phi_def}
    \begin{split}
    \phi(x,h) &\dfn \frac{\Pi'\Psi(x)}{\alpha_I +\alpha_N} + \frac{1}{2}(x-V(h))'(P_I-P_U)(x-V(h)),\\
    V(h) &\dfn (P_I - P_U)^{-1}\left(\left(\frac{\alpha_I + \alpha_N\tau_N}{\alpha_I + \alpha_N}P_I - P_U\right)h + \frac{\alpha_N}{\alpha_I +\alpha_N} P_I \mu_N\right).
    \end{split}
\end{equation}
These forms will help us deduce properties of the solution $z$. In particular, by Assumption \ref{A:Psi}
\begin{equation}\label{E:vartheta_to_V_ordering}
    \begin{split}
        \vartheta(x,h) &=(\alpha_I +\alpha_N)\phi(x,h) - \frac{1}{2}(\alpha_I +\alpha_N)V(h)'(P_I - P_U)V(h),\\
        &\geq - \frac{1}{2}(\alpha_I +\alpha_N)V(h)'(P_I - P_U)V(h).
    \end{split}
\end{equation}
\end{rem}


Our first lemma establishes additional properties of the solution $z$ to \eqref{E:clear_function_00} from Lemma \ref{L:z_soln}. To state it, define the absolute risk aversion  and weighted risk tolerance functions
\begin{equation*}
\gamma_U(w) \dfn -\frac{\ddot{U}(w)}{\dot{U}(w)}, \qquad \alpha_U(w) \dfn \frac{\omega_U}{\gamma_U(w)}.
\end{equation*}

\begin{lem}\label{L:clear_function_soln}
For each fixed $(x,h,\kappa)$ let $z = z(x,h,\kappa)$ be the unique solution to \eqref{E:clear_function_00} from Lemma \ref{L:z_soln}. Then, $z$ is strictly decreasing in $\kappa$ with derivative
\begin{equation}\label{E:partial_kappa_new}
    \frac{\partial_{\kappa} z}{z}= \frac{-1}{\alpha_U(\imutil(z)) + \alpha_I +\alpha_N}.
\end{equation}
Next, recall $\phi,V$ from \eqref{E:phi_def} and let $\wt{z}(h,\kappa)$ denote the unique solution to
\begin{equation}\label{E:clear_no_x}
    \omega_U \imutil(z) - (\alpha_I + \alpha_N)\log(z) = -\frac{1}{2}(\alpha_I + \alpha_N)V(h)'(P_I-P_U)V(h) + \kappa,
\end{equation}
Then
\begin{equation}\label{E:bounds_1}
    \wt{z}(h,\kappa) e^{-\phi(x,h)} \leq z(x,h,\kappa) \leq \wt{z}(h,\kappa),
\end{equation}
as well as
\begin{equation}\label{E:bounds_2}
    \imutil(z(x,h,\kappa)) \leq \imutil(\wt{z}(h,\kappa)) +  \frac{\alpha_I + \alpha_N}{\omega_U}\phi(x,h).
\end{equation}

\end{lem}

\begin{proof}[Proof of Lemma \ref{L:clear_function_soln}]
That \eqref{E:partial_kappa_new} holds is clear from \eqref{E:clear_function_00} and $z = \dot{\util}(\imutil(z))$, and the derivative is negative because $\alpha_U > 0$. The upper bound in \eqref{E:bounds_1} is clear from \eqref{E:vartheta_to_V_ordering} as $z\to \omega_U \imutil(z) - (\alpha_I + \alpha_N)\log(z)$  is decreasing.  The lower bound in \eqref{E:bounds_1} follows because $\imutil$ is decreasing and $z(x,h,\kappa)\leq \wt{z}(h,\kappa)$ so that
\begin{equation*}
\begin{split}
    - (\alpha_I + \alpha_N)\log\left(\frac{z(x,h,\kappa)}{\wt{z}(h,\kappa)}\right) &\leq \omega_U(\imutil(z(x,h,\kappa)) - \imutil(\wt{z}(h,\kappa))) - (\alpha_I + \alpha_N)\log\left(\frac{z(x,h,\kappa)}{\wt{z}(h,\kappa)}\right)\\
    &= (\alpha_I+\alpha_N)\phi(x,h).
\end{split}
\end{equation*}
Similarly, \eqref{E:bounds_2} holds because
\begin{equation*}
\begin{split}
    \omega_U(\imutil(z(x,h,\kappa)) - \imutil(\wt{z}(h,\kappa))) &\leq \omega_U(\imutil(z(x,h,\kappa)) - \imutil(\wt{z}(\kappa))) - (\alpha_I + \alpha_N)\log\left(\frac{z(x,h,\kappa)}{\wt{z}(h,\kappa)}\right)\\
    &= (\alpha_I+\alpha_N)\phi(x,h).
\end{split}
\end{equation*}

\end{proof}

Having obtained solutions to \eqref{E:clear_function_00}, let us now consider \eqref{E:U_budget_new} and \eqref{E:clear_real_2}. Using Assumptions \ref{A:Psi}, \ref{A:init_pos} with \eqref{E:ell_T_def}, \eqref{E:qprobm_to_qprobh}, and Lemma \ref{L:clear_function_soln} it is clear that for all $(h,\kappa)$
\begin{equation*}
    \begin{split}
    \expvs{\max\bra{1,\imutil(z(X_T,h,\kappa)),(\pi^U_0)'\Psi(X_T),|\chi(X_T,h)|}z(X_T,h,\kappa)e^{-\frac{1}{2}X_T'P_U X_T + X_T'P_U h}} < \infty,\\
    \expvs{\log(z(X_T,h,\kappa))z(X_T,h,\kappa)e^{-\frac{1}{2}X_T'P_U X_T + X_T'P_U h}} < \infty,
    \end{split}
\end{equation*}
and hence the expectations in \eqref{E:U_budget_new} and \eqref{E:clear_real_2} are well defined. The next lemma shows for any $(h,\kappa)$ that if $z$ solves \eqref{E:clear_function_00} then \eqref{E:clear_real_2} holds.

\begin{lem}\label{L:clear_gen_new}
For fixed $(h,\kappa)$ if $z$ solves \eqref{E:clear_function_00} then \eqref{E:clear_real_2} holds.
\end{lem}

\begin{proof}[Proof of Lemma \ref{L:clear_gen_new}]
Clearly, $z$ is Borel measurable in $x$.  Next, for fixed $(h,\kappa)$, if $z$ solves \eqref{E:clear_function_00} then
\begin{equation*}
    \begin{split}
        &\omega_U\left(\imutil\left(z\right)-\frac{\expvs{\imutil(z) z\ell}}{\expvs{z\ell}}\right) - (\alpha_I + \alpha_N)\left(\log(z) - \frac{\expvs{\log(z)z\ell}}{\expvs{z\ell}}\right)\\
        &= \left(\omega_U \imutil(z) - (\alpha_I + \alpha_N)\log(z)\right) - \frac{\expvs{\left(\omega_U \imutil(z) - (\alpha_I + \alpha_N)\log(z)\right)z\ell}}{\expvs{z\ell}},\\
        &= \vartheta  - \frac{\expvs{\vartheta z\ell}}{\expvs{z\ell}}.
    \end{split}
\end{equation*}
Above, we used \eqref{E:clear_function_00} and that $\kappa$ is constant in the third  equality. This gives \eqref{E:clear_real_2}.

\end{proof}

In light of Lemma \ref{L:clear_gen_new},Theorem \ref{T:wU0_exists_gen} will follow if we identify a map $h \to \kappa(h)$ such that \eqref{E:U_budget_new} holds for $z$ solving \eqref{E:clear_function_00}. This is done below, and constitutes the proof of Theorem  \ref{T:wU0_exists_gen}.

\begin{proof}[Proof of Theorem \ref{T:wU0_exists_gen}] Omitting function arguments, we seek $\kappa = \kappa(h)$ enforcing \eqref{E:U_budget_new}, which by using \eqref{E:ell_T_def} is equivalent to
\begin{equation}\label{E:key_fraction}
    1 = g(\kappa) \dfn \frac{\expvs{\Psi'\pi^U_0\ell z}}{\expvs{\imutil(z)\ell z }} = \frac{\expvs{\Psi'\pi^U_0 e^{\frac{1}{\alpha_I+\alpha_N}\left(\omega_U \imutil(z) -\chi\right)}}}{\expvs{\imutil(z) e^{\frac{1}{\alpha_I+\alpha_N}\left(\omega_U \imutil(z) -\chi\right)}}},
\end{equation}
where the second equality holds because $z$ solves  \eqref{E:clear_function}.  To prove existence of such $\kappa$ we will show that $g$ is continuous in $\kappa$ with  $\lim_{\kappa\to -\infty} g(\kappa) = \infty$ and $\lim_{\kappa\to \infty} g(\kappa) = 0$.  We will also show that provided $w\to w\dot{\util}(w)$ is non-deceasing then $g$ is strictly decreasing in $\kappa$ and solutions are unique. Lastly, we will prove the measurability results.

We start with the limiting statements.  Let $\kappa \to -\infty$, which in view of Lemma \ref{L:clear_function_soln} implies the almost sure limits $z(X_T,h,\kappa)\uparrow \infty, \imutil(z(X_T,h,\kappa)) \downarrow 0$. This means that for $\kappa < 0$ we have almost surely that (here, we omit the $(X_T,h)$ function arguments for brevity, and note in \eqref{E:clear_function} that $\kappa_1 = \kappa - \log(p_{C_U})$)
\begin{equation*}
    \begin{split}
    \imutil(z(\kappa))e^{\frac{1}{\alpha_I + \alpha_N}\left(\omega_U \imutil(z(\kappa)) -\chi\right)} &\leq \imutil(z(0))e^{\frac{1}{\alpha_I + \alpha_N}\left(\omega_U \imutil(z(0)) -\chi\right)} = \imutil(z(0))z(0)\ell p_{C_U}^{-\frac{1}{\alpha_I +\alpha_N}}.
    \end{split}
\end{equation*}
By \eqref{E:bounds_2} we know the right-most quantity above is integrable, and hence by dominated convergence the denominator in \eqref{E:key_fraction} goes to $0$. As for the numerator, recall from Assumption \ref{A:init_pos} that $(\pi^U_0)'\Psi(X_T) > 0$. Thus,  by Fatou's lemma, the numerator exceeds (in the limit)
\begin{equation*}
    \expvs{\Psi'\pi^U_0 e^{-\chi/(\alpha_I + \alpha_N)}} > 0,
\end{equation*}
so that $\lim_{\kappa\to -\infty} g(\kappa) = \infty$.  As for when $\kappa \to \infty$, recalling $\phi$ from \eqref{E:phi_def} we have by \eqref{E:bounds_1}, \eqref{E:bounds_2}  and because $\imutil$ is decreasing
\begin{equation*}
    g(\kappa) \leq  \frac{\expvs{\Psi'\pi^U_0 \ell}}{\imutil(\wt{z}(\kappa))\expvs{\ell e^{-\phi}}}.
\end{equation*}
The expectations on the right side no longer depend on $\kappa$ and each are finite.  Thus, $\lim_{\kappa \to \infty} g(\kappa) = 0$, because  $\imutil(\wt{z}(\kappa)) \to \infty$.

We next prove continuity by showing the numerator and denominator in \eqref{E:key_fraction} are each continuous.  For the numerator, note that $\Psi'\pi^U_0$ does not depend on $\kappa$, and from Lemma \ref{L:clear_function_soln} we have for $1 >\eps > 0$
\begin{equation*}
    \begin{split}
        0 &\geq \frac{1}{\eps}(z(x,h,\kappa+\eps) - z(x,h,\kappa))  = \frac{1}{\eps}\int_{\kappa}^{\kappa+\eps} \partial_{\kappa} z(x,h,\tau) d\tau\\
        &= -\frac{1}{\eps}\int_{\kappa}^{\kappa+\eps} \frac{z(x,h,\tau)}{\alpha_U(\imutil(z(x,h,\tau))) + \alpha_I+\alpha_N}d\tau,\\
        &\geq -\frac{1}{(\alpha_I+\alpha_N)\eps}\int_{\kappa}^{\kappa+\eps} z(x,h,\tau)d\tau \geq -\frac{z(x,h,\kappa)}{(\alpha_I+\alpha_N)} \geq -\frac{\wt{z}(h,\kappa)}{\alpha_I +\alpha_N}.
    \end{split}
\end{equation*}
This implies
\begin{equation*}
    0 \geq \Psi(X_T)'\pi^U_0 \ell(T,X_T,h) \frac{1}{\eps}(z(X_T,h,\kappa+\eps) - z(X_T,h,\kappa)) \geq - \Psi(X_T)'\pi^U_0 \ell(T,X_T,h)\frac{\wt{z}(h,\kappa)}{\alpha_I +\alpha_N}
\end{equation*}
Therefore, by Assumption \ref{A:Psi} and \eqref{E:ell_T_def} we see that the numerator is not only continuous, but differentiable in $\kappa$. Turning to the denominator,
using \eqref{E:bounds_2} we obtain for $0<\eps < 1$ (as $\kappa \to \imutil(z(x,h,\kappa))$ is increasing)
\begin{equation*}
    \begin{split}
        0 &\geq \imutil(z(x,h,\kappa+\eps))\frac{z(x,h,\kappa+\eps) - z(x,h,\kappa)}{\eps}  \geq -\frac{\wt{z}(h,\kappa)\imutil(z(x,h,\kappa+\eps))}{\alpha_I +\alpha_N},\\
        &\geq -\frac{\wt{z}(h,\kappa)}{\alpha_I+\alpha_N}\left(\imutil(\wt{z}(h,\kappa+1)) +  \frac{\alpha_I +\alpha_N}{\omega_U}\phi(X_T,h)\right),
    \end{split}
\end{equation*}
and
\begin{equation*}
    \begin{split}
        0 &\leq z(x,h,\kappa)\frac{\imutil(z(x,h,\kappa+\eps)) - \imutil(z(x,h,\kappa))}{\eps}  = \frac{z(x,h,\kappa)}{\eps}\int_{\kappa}^{\kappa+\eps} \frac{-\dot{\imutil}(z(x,h,\tau))z(x,h,\tau)}{\alpha_U( \imutil(z(x,h,\tau))) + \alpha_I+\alpha_N}d\tau,\\
        &= \frac{z(x,h,\kappa)}{\eps\omega_U}\int_{\kappa}^{\kappa+\eps} \frac{\alpha_U( \imutil(z(x,h,\tau)))}{\alpha_U( \imutil(z(x,h,\tau))) + \alpha_I+\alpha_N}d\tau,\\
        &\leq \frac{z(x,h,\kappa)}{\omega_U} \leq \frac{\wt{z}(h,\kappa)}{\omega_U}.
    \end{split}
\end{equation*}
Again, using $\ell$ from \eqref{E:ell_T_def} we find that the denominator is also differentiable, hence continuous in $\kappa$.  Therefore, by continuity there exists a solution $\wh{\kappa}$ to $g(\kappa) = 1$. As for the uniqueness statement, clearly $\kappa \to \expvs{\Psi'\pi^U_0 \ell z}$ is strictly decreasing.  Additionally,
\begin{equation*}
    \frac{d}{dz} (z\imutil(z)) = \frac{1}{\ddot{\util}(w)}\frac{d}{dw} (w \dot{\util}(w));\qquad w = \imutil(z).
\end{equation*}
Thus, if $w\to w\dot{\util}(w)$ is non-decreasing then $z \imutil(z)$ is non-increasing in $z$ and hence non-increasing (almost surely) in $\kappa$. Therefore, $\kappa \to g(\kappa)$ is strictly decreasing and hence uniqueness follows.

The last thing to prove is measurability of $\wh{\kappa}$ in $h$, which follows from  the Kuratowski-Ryll-Nardzewski measurable selection theorem. Indeed, writing $g$ as a function of both $(h,\kappa)$ for $E\in\mathcal{B}(\reals)$ define
\begin{equation*}
    A^E \dfn \cbra{h\in\reals^d \such \ \exists\ \kappa \in E \textrm{ s.t. } g(h,\kappa) = 1}.
\end{equation*}
Assuming $g$ is jointly continuous in $(h,\kappa)$ (which can be shown using Lemma \ref{L:clear_function_soln}) it is easy to see that $A^E$ is closed (hence Borel) for each closed $E$.  Now, let $E = (a,b)$ for $a<b$ and (for $n$ large enough so that it is not empty) set $E_n = [a+1/n,b-1/n]$.  Straightforward computations show that
\begin{equation*}
    A^E = \bigcup_n A^{E_n}
\end{equation*}
and hence $A^E$ is Borel for each $E=(a,b)$.  But this implies that $A^E$ is Borel for each open $E$ and hence the Kuratowski-Ryll-Nardzewski measurable selection theorem implies that we may select $\wh{\kappa}$ measurably in $h$.

\end{proof}


\section{Proofs from Section \ref{S:PCE}}\label{AS:PCE}

\begin{proof}[Proof of Lemma \ref{L:dual_realization_ok}]
From \eqref{E:gen_cz_nice} we see that
\begin{equation*}
    Z^h_T \dfn \frac{d\qprob^h}{d\wh{\prob}^h}\big|_{\F^B_T} = \frac{z(X_T,h)}{\expv{\wh{\prob}^h}{}{z(X_T,h)}}.
\end{equation*}
We first claim that \eqref{E:dual_vf_cond} holds when $\lambda \geq \expv{\wh{\prob}^h}{}{z(X_T,h)}$. Indeed, as $\dot{\dualutil} = -\imutil < 0$,  $\dualutil$ is decreasing and convex. Therefore,
\begin{equation*}
    \expv{\wh{\prob}^h}{}{\dualutil\left(\lambda \frac{d\qprob^h}{d\wh{\prob}^h}\big|_{\F^B_T}\right)} \leq \expv{\wh{\prob}^h}{}{\dualutil(z(X_T,h))}.
\end{equation*}
Next, for any $z>0$
\begin{equation*}
    \dualutil(z) \leq |\dualutil(1)| + \dualutil(z)1_{z<1} \leq 2|\dualutil(1)| + (1-z)\imutil(z)1_{z<1}.
\end{equation*}
Using \eqref{E:bounds_2} in Lemma \ref{L:clear_function_soln} we deduce
\begin{equation*}
    \begin{split}
    \imutil(z) &\leq C(h)\left(1 + |x|^2 + \Pi'\Psi(x)\right)
    \end{split}
\end{equation*}
for some constant $C(h)>0$. Therefore, for $\lambda \geq \expv{\wh{\prob}^h}{}{z(X_T,h)}$
\begin{equation*}
    \begin{split}
     \expv{\wh{\prob}^h}{}{\dualutil\left(\lambda \frac{d\qprob^h}{d\wh{\prob}^h}\big|_{\F^B_T}\right)} &\leq \expv{\wh{\prob}^h}{}{\dualutil(z(X_T,h))},\\
     &\leq 2|\dualutil(1)| + C(h)\expv{\wh{\prob}^h}{}{(1-z(X_T,h))1_{z(X_T,h)\leq 1}\left(1 + |X_T|^2 + \Pi'\Psi(X_T)\right)},\\
     &\leq 2|\dualutil(1)| + C(h)\expv{\wh{\prob}^h}{}{\left(1 + |X_T|^2 + \Pi'\Psi(X_T)\right)},\\
    \end{split}
\end{equation*}
The result follows from Assumption \ref{A:Psi} and the construction of $\wh{\prob}^h$.  When $0 < \lambda < \expv{\wh{\prob}^h}{}{z(X_T,h)}$ we have
\begin{equation*}
    \expv{\wh{\prob}^h}{}{\dualutil\left(\lambda \frac{d\qprob^h}{d\wh{\prob}^h}\big|_{\F^B_T}\right)} =  \expv{\wh{\prob}^h}{}{\dualutil\left(\wt{\lambda}z(X_T,h)\right)},\quad \wt{\lambda} = \frac{\lambda}{\expv{\wh{\prob}^h}{}{z(X_T,h)}} < 1.
\end{equation*}
As $U$ satisfies the reasonable asymptotic elasticity condition in Assumption \ref{A:U}, we deduce from \cite[Corollary 6.1 (iii')]{MR2023886} the existence of $z_0 > 0$ and a constant $C(\wt{\lambda})$ such that $\dualutil(\wt{\lambda}z) \leq C(\wt{\lambda})\dualutil(z)$ when $z < z_0$.  Using again that $\dualutil$ is decreasing, this implies
\begin{equation*}
    \begin{split}
        \dualutil(\wt{\lambda}z) &= \dualutil(\wt{\lambda}z) 1_{z\geq z_0} + \dualutil(\wt{\lambda}z)1_{z < z_0} \leq |\dualutil(\wt{\lambda} z_0)| + C(\wt{\lambda})\dualutil(z)1_{z < z_0}.
    \end{split}
\end{equation*}
By definition of $\dualutil$ we know $\dualutil(z) \geq \util(1) - z$ so that
\begin{equation*}
    \dualutil(z)1_{z < z_0} = \dualutil(z) - \dualutil(z)1_{z\geq z_0} \leq \dualutil(z) + |\util(1)| + z.
\end{equation*}
Therefore,
\begin{equation*}
    \begin{split}
        \dualutil(\wt{\lambda}z(Z_T,h))\leq |\dualutil(\wt{\lambda} z_0)| + C(\wt{\lambda})\left(|\util(1)|  + \dualutil(z(X_T,h)) + z(X_T,h)\right).
    \end{split}
\end{equation*}
\eqref{E:dual_vf_cond} holds because we already showed $\expv{\wh{\prob}^h}{}{\dualutil(z(X_T,h))} < \infty$ and by Lemma \ref{L:clear_function_soln} we know $z(X_T,h)\leq C(h)$ for some constant $C$.
\end{proof}

\section{Proofs from Section \ref{S:DNREE}}\label{AS:DNREE}

\begin{proof}[Proof of Proposition \ref{P:DNREE_1}]
Define
\begin{equation*}
    \Lambda(t,x,h) \dfn \log\left(\condexpvs{z(X_T,h,\wh{\kappa}(h))e^{-\frac{1}{2}X_T'P_U X_T + X_T'P_U h}}{X_t=x}\right).
\end{equation*}
Under the given assumptions, we know $\Lambda$ is smooth enough in $(t,x)$ to use Ito's formula. Therefore, from  \eqref{E:price_def}, \eqref{E:qprobm_to_qprobh}, \eqref{E:gen_cz_new} we see that
\begin{equation*}
    \frac{d\qprob^h}{d\prob}\big|_{\F^B_T} = e^{\Lambda(T,X_T,h) - \Lambda(0,X_0,h)} = \E\left(\int_0^\cdot \partial_x \Lambda(t,X_t,h)'a(X_t)dB_t\right)_T.
\end{equation*}
Writing the price process $t \to S(t,X_t,h)$ as $S^h$ we see that  $S^h$ has market price of risk $\nu^h_t \dfn -a(X_t)' \partial_x \Lambda(t,X_t,h)$. 
Next, write $L$ as the second order operator associated to $X$, so that for smooth functions $f$
\begin{equation*}
    (Lf)[x] = \frac{1}{2}A(x)\partial_{xx} f(x) + b(x)\partial_x f(x).
\end{equation*}
We know that $S(\cdot,h)$ solves the PDE
\begin{equation*}
    \partial_t u + Lu + \partial_x u A \partial_x \Lambda(\cdot,h)  =0;\qquad u(T,\cdot) = \Psi.
\end{equation*}
Therefore, if the statement is true (which implies $\partial_t, \partial_x, \partial_{xx}$ coincide for $S(\cdot, h_1), S(\cdot,h_2)$) then it follows that on $(0,\eps)\times\reals$ we have
\begin{equation*}
    0 = (a(x)\partial_x S(t,x,h_i))a(x)\left(\partial_x\Lambda(t,x,h_1) - \partial_x\Lambda(t,x,h_2)\right);\qquad i = 1,2.
\end{equation*}
As $\sigma(t,x,h_i) = a(x)\partial_x S(t,x,h_i)$ we know from Proposition \ref{P:on_TA_complete} that for each $h$, and (Lebesgue almost every) $(t,x)$ that $\sigma(t,x,h) \neq 0$.  Thus, as $a(x)$ is non-degenerate as well, we have Lebesgue almost surely on $(0,\eps)\times \reals$ that
\begin{equation}\label{E:Lambda1}
    \partial_x \Lambda(t,x,h_1) =  \partial_x \Lambda(t,x,h_2).
\end{equation}
By continuity the above holds for all $(t,x)\in (0,\eps)\times \reals$. Thus, as $\Lambda(\cdot,h)$ solves the PDE
\begin{equation*}
    \partial_t u + Lu + \frac{1}{2}\partial_x u A \partial_x u  =0;\qquad u(T,x) = \log(z(x, h,\wh{\kappa}(h))) - \frac{1}{2}x'P_U x + x'P_U h.
\end{equation*}
we see that $v \dfn \Lambda(\cdot,h_1) - \Lambda(\cdot,h_2)$ satisfies, on $(0,\eps)\times\reals$
\begin{equation*}
    \partial_t v + Lv = 0.
\end{equation*}
Again, using \eqref{E:Lambda1} we know by Ito's formula that almost surely
\begin{equation*}
    \Lambda(t,X_t,h_1) - \Lambda(0,X_0,h_1) = \Lambda(t,X_t,h_2) - \Lambda(0,X_0,h_2)
\end{equation*}
which by the support theorem for diffusion processes (see \cite{MR0400425}) implies for Lebesgue almost every $(0,\eps)\times\reals$ (hence all by continuity) that
\begin{equation*}
    \Lambda(t,x,h_1) - \Lambda(0,X_0,h_1) = \Lambda(t,x,h_2) - \Lambda(0,X_0,h_2)
\end{equation*}
But, this implies
\begin{equation*}
    w(t,x) \dfn e^{\Lambda(t,x,h_1)-\Lambda(0,X_0,h_1)} - e^{\Lambda(t,x,h_2) - \Lambda(0,X_0,h_2)}
\end{equation*}
satisfies the PDE
\begin{equation*}
    \begin{split}
    \partial_t w + Lw &= 0;\\
        w(T,x) &= z(x,h_1,\wh{\kappa}(h_1))e^{ - \frac{1}{2}x'P_U x + x'P_U h_1-\Lambda(0,X_0,h_1)}\\
        &\qquad\qquad - z(x,h_2,\wh{\kappa}(h_2))e^{ - \frac{1}{2}x'P_U x + x'P_U h_2-\Lambda(0,X_0,h_2)},
    \end{split}
\end{equation*}
and is such that $w(t,x) \equiv 0$ almost surely on $(0,\eps)\times \reals$. We now apply \cite[Theorem 1.2]{MR3989959} to $w$. To match their notation define the operator $P = \partial_t + (1/2)\partial_x (A \partial_x)$. We then have
\begin{equation*}
    P w = \frac{1}{2} \partial_x A \partial_x w - b \partial_x w,
\end{equation*}
so that $|Pw|\leq N| \partial_x w|$ for a constant $N$. As Lemma \ref{L:clear_function_soln} implies $|w| \leq C$ for some $C(h)$, the assumptions of \cite[Theorem 1.2]{MR3989959} are met with $\alpha = 1$ therein and we may conclude that if there exists any $t_0$ such that $w(t_0,\cdot)$ is identically zero, then it must be that $w(T,\cdot)$ is identically constant as well.  But, this is an obvious contradiction if $h_1 \neq h_2$.
\end{proof}



\section{Proofs from Section \ref{S:asympt}}\label{AS:asympt}

\begin{proof}[Proof of Proposition \ref{P:large_etaU}]

Throughout, $h$ is fixed.  Let us first assume there is a sequence $\eta_n \to \infty$ and a number $\wt{\kappa}$ such that
\begin{equation}\label{E:large_kappa_conv}
    \lim_{n\to\infty} \wt{\kappa}_n \dfn \frac{\wh{\kappa}(h,\eta_n)}{(\alpha_I + \alpha_U)\eta_n} = \wt{\kappa}.
\end{equation}
Given this, we claim
\begin{equation}\label{E:large_eta_zlim}
    \begin{split}
    \lim_{n\to\infty} z(x,h,\wh{\kappa}(h,\eta_n)) e^{\eta_n \wt{\kappa}_n} &= e^{Ae^{\wt{\kappa}}- B(x,h)};\qquad A = \frac{\omega_U}{\alpha_I+\alpha_N}, \ B(x,h) = \frac{\vartheta(x,h)}{\alpha_I + \alpha_N}.
    \end{split}
\end{equation}
First, \eqref{E:z_crra} implies
\begin{equation}\label{E:large_temp_z}
    z(x,h,\wh{\kappa}(h,\eta_n)) = \left(\frac{A}{\eta_n}\right)^{\eta_n} \prodlog{\frac{A}{\eta_n} e^{\frac{B(x,h)}{\eta_n}+\wt{\kappa}_n}}^{-\eta_n}.
\end{equation}
Next, for any $k$, $\prodlog{z}^{-k} = z^{-k}e^{k\prodlog{z}}$ and hence
\begin{equation}\label{E:large_temper}
    z(x,h,\wh{\kappa}(h,\eta_n))e^{\eta_n\wt{\kappa}_n} = e^{-B(x,h)  + \eta_n \prodlog{\frac{A}{\eta_n} e^{\frac{B(x,h)}{\eta_n} + \wt{\kappa}_n}}}.
\end{equation}
Using  $\lim_{y\to 0} \prodlog{y}/y = 1$ we obtain for any $A>0,B\in\reals$
\begin{equation}\label{E:PL_ident}
\lim_{y \to 0} \frac{\prodlog{A y e^{y B+ \wt{\kappa}}}}{y}  = \lim_{y \to 0} \frac{\prodlog{A y e^{y B + \wt{\kappa}}}}{Aye^{By + \wt{\kappa}}}Ae^{B y + \wt{\kappa}} =  A e^{\wt{\kappa}},
\end{equation}
Using this, and the assumption that $\wt{\kappa}_n \to \wt{\kappa}$ we verify \eqref{E:large_eta_zlim}. We now turn to identifying $\wt{\kappa}$.  To this end, re-writing \eqref{E:U_budget_new}, using \eqref{E:ell_T_def},  and setting $\wh{\kappa}_n = \wh{\kappa}(h,\eta_n)$, we have for each $n$
\begin{equation}\label{E:large_kappa_map}
    1 = \frac{\expvs{(\pi_0^U)'\Psi(X_T) \ell(T,X_T,h) z(X_T,h,\wh{\kappa}_n)}}{\expvs{\ell(T,X_T,h) z(X_T,h,\wh{\kappa}_n)^{1-\frac{1}{\eta_n}}}}.
\end{equation}
From \eqref{E:large_temp_z} we find
\begin{equation*}
    z(x,h,\wh{\kappa}(h,\eta_{n}))e^{\eta_n\wt{\kappa}_n} = \left(\frac{Ae^{\wt{\kappa}_n}}{\eta_n}\right)^{\eta_n}\prodlog{\frac{A}{\eta_n} e^{\frac{B(x,h)}{\eta_n}+\wt{\kappa}_n}}^{-\eta_n},
\end{equation*}
From \eqref{E:vartheta_to_V_ordering} we know (as $h$ is fixed) that $B(x,h)$ is bounded from below by some constant $\ul{B}(h)$. As $\prodlog{\cdot}$ is increasing this implies
\begin{equation*}
    \begin{split}
    z(x,h,\wh{\kappa}(h,\eta_{n}))e^{\eta_n\wt{\kappa}_n} &\leq  \left(\frac{Ae^{\wt{\kappa}_n}}{\eta_n}\right)^{\eta_n}\prodlog{\frac{A}{\eta_n} e^{\frac{\ul{B}(h)}{\eta_n}+\wt{\kappa}_n}}^{-\eta_n}= \left(\frac{\frac{A}{\eta_n}e^{\frac{\ul{B}(h)}{\eta_n} + \wt{\kappa}_n}}{\prodlog{\frac{A}{\eta_n} e^{\frac{\ul{B}(h)}{\eta_n}+\wt{\kappa}_n}}}\right)^{\eta_n}e^{-\ul{B}(h)};\\
    &= e^{-\ul{B(h)} + \eta_n \prodlog{\frac{A}{\eta_n} e^{\frac{\ul{B}(h)}{\eta_n}+\wt{\kappa}_n}}},
    \end{split}
\end{equation*}
where the last equality hold as $z/\prodlog{z} = e^{\prodlog{z}}$. $\wt{\kappa}_n \to \wt{\kappa}$ and \eqref{E:PL_ident} then imply that for $n$ large enough
\begin{equation}\label{E:large_temper_2}
    z(x,h,\wh{\kappa}(h,\eta_{n}))e^{\eta_n\wt{\kappa}_n} \leq  e^{-\ul{B}(h) + A e^{\wt{\kappa}} + 1},
\end{equation}
and hence $z(X_T,h,\wh{\kappa}(h,\eta_{n}))e^{\eta_n\wt{\kappa}_n}$ is almost surely bounded from above.  We may express \eqref{E:large_kappa_map} as
\begin{equation}\label{E:temper_5}
    \begin{split}
        1  = \frac{\expvs{(\pi_0^U)'\Psi(X_T) \ell(T,X_T,h) z(X_T,h,\wh{\kappa}_n)e^{\eta_n\wt{\kappa}_n}}}{e^{\wt{\kappa}_n}\expvs{\ell(T,X_T,h) \left(z(X_T,h,\wh{\kappa}_n)e^{\eta_n\wt{\kappa}_n}\right)^{1-\frac{1}{\eta_n}}}}.
    \end{split}
\end{equation}
The bounded convergence theorem and \eqref{E:large_eta_zlim} yield
\begin{equation*}
    \wt{\kappa} = \log\left(\frac{\expvs{(\pi_0^U)'\Psi(X_T) \ell(T,X_T,h) e^{-B(X_T,h)}}}{\expvs{\ell(T,X_T,h) e^{-B(X_T,h)}}}\right).
\end{equation*}
and hence for any $\eta_n \to \infty$ such that $\wh{\kappa}_n$ from \eqref{E:large_kappa_conv} converges, the limit is uniquely identified. Lastly, note that
\begin{equation*}
    S(t,x,h;\eta_n) = \frac{\condexpvs{\Psi(X_T) \ell(T,X_T,h) z(X_T,h,\wh{\kappa}_n)}{X_t = x}}{\condexpvs{\ell(T,X_T,h)z(X_T,h,\wh{\kappa}_n)}{X_t=x}}.
\end{equation*}
Using the exact same arguments as above for $(t,x,h)$ fixed (i.e. multiplying the numerator and denominator by $e^{\eta_n\wt{\kappa}_n}$ and using the bounded convergence theorem), we may pass the limit through the expectation operator to obtain \eqref{E:large_etaU_price}.

The final item to prove is that $\eta_n \to \infty$ implies $\wt{\kappa}_n$ from \eqref{E:large_kappa_conv} converges to a finite (unique) limit, for whatever $\wh{\kappa}(h,\eta_n)$ we choose that satisfies \eqref{E:U_budget_new}.  As we have already shown that every sub-sequence $\cbra{\wt{\kappa}_{n_k}}$ which converges to a finite limit, has the same limit point, it suffices to show $\wt{\kappa}_n \to \pm \infty$ cannot happen.

Assume $\wt{\kappa}_n \to -\infty$.  From \eqref{E:large_temper} and \eqref{E:large_kappa_map} we obtain
\begin{equation*}
    e^{\wt{\kappa}_n} = \frac{\expvs{(\pi_0^U)'\Psi(X_T) \ell(T,X_T,h) e^{-B(X_T,h) + \eta_n\prodlog{\frac{A}{\eta_n} e^{\frac{B(x,h)}{\eta_n}+\wt{\kappa}_n}}}}}{\expvs{\ell(T,X_T,h) e^{\left(1-\frac{1}{\eta_n}\right)\left(-B(X_T,h) + \eta_n\prodlog{\frac{A}{\eta_n} e^{\frac{B(x,h)}{\eta_n}+\wt{\kappa}_n}}\right)}}} \rdfn \frac{N_n}{D_n}.
\end{equation*}
Using that $\prodlog{\cdot}$ is increasing, \eqref{E:PL_ident} and Fatou's lemma, we deduce
\begin{equation*}
    \liminf_{n\to\infty} N_n \geq \expvs{(\pi_0^U)'\Psi(X_T) \ell(T,X_T,h) e^{-B(X_T,h)}}.
\end{equation*}
As for the denominator, let $\wt{\kappa} > 0$ be fixed and take $n$ large enough so that $\wt{\kappa}_n \leq -\wt{\kappa}$ for $n$ large enough. As $\prodlog{\cdot}$ is increasing and $\eta_n$ is large we have
\begin{equation*}
    \begin{split}
    D_n &\leq \expvs{\ell(T,X_T,h) e^{\left(1-\frac{1}{\eta_n}\right)\left(-B(X_T,h) + \eta_n\prodlog{\frac{A}{\eta_n} e^{\frac{B(x,h)}{\eta_n}-\wt{\kappa}}}\right)}},\\
    &= \expvs{\ell(T,X_T,h) \left(z(x,h,\wt{\kappa})e^{\eta_n\wt{\kappa}}\right)^{1-\frac{1}{\eta_n}}},\\
    &\leq e^{\left(1-\frac{1}{\eta_n}\right)\left(-\ul{B}(h) + Ae^{-\wt{\kappa}}+1\right)}\expvs{\ell(T,X_T,h)}.
    \end{split}
\end{equation*}
Above, the second equality follows from \eqref{E:large_temper} and the inequality from \eqref{E:large_temper_2}, which was valid for general constants $\wt{\kappa}_n \equiv \wt{\kappa}$ and does not require \eqref{E:large_kappa_map}.  As such
\begin{equation*}
    \limsup_{n} D_n \leq e^{\left(-\ul{B}(h) + Ae^{-\wt{\kappa}}+1\right)}\expvs{\ell(T,X_T,h)}.
\end{equation*}
As such, $\wt{\kappa}_n \to -\infty$ contradicts \eqref{E:large_kappa_map} in the limit $n\to\infty$.   Lastly, assume $\wt{\kappa}_n \to \infty$. Here, we will use a slightly different argument to reach a contradiction. First, using that $\prodlog{\cdot}$ is increasing and \eqref{E:PL_ident},  if $\wt{\kappa}_n \to \infty$ then for any constants $A>0,\check{B}$
\begin{equation}\label{E:larget_temper_4}
    \lim_{n\to\infty} \eta_n \prodlog{\frac{A}{\eta_n}e^{\frac{\check{B}}{\eta_n} + \wt{\kappa}_n}} = \infty.
\end{equation}
Next, from \eqref{E:bounds_1}, \eqref{E:temper_5}, and using the notation of \eqref{E:phi_def} we deduce that
\begin{equation*}
    1  \leq \frac{\expvs{(\pi^U_0)'\Psi(X_T)\ell(T,X_T,h)}}{\expvs{\ell(T,X_T,H)e^{-\phi(X_T,h)}}} \times \check{z}(h,\wh{\kappa}(h,\eta_n))^{\frac{1}{\eta_n}}.
\end{equation*}
Above, we have written $\check{z}(h,\wh{\kappa}(h,\eta_n))$ for the quantity in \eqref{E:bounds_1}, and used Assumption \ref{A:Psi} which implies that $\phi$ is non-negative.  Therefore, there is a quantity $C(h)$ depending only on $h$ (and not $\wt{\kappa}_n,\eta_n$) such that
\begin{equation*}
     \check{z}(h,\wh{\kappa}(h,\eta_n)) \geq  C(h)^{\eta_n}.
\end{equation*}
From \eqref{E:clear_no_x}, and using $I(y) = y^{-1/\eta_n}$ we have (see \eqref{E:z_crra})
\begin{equation*}
    \check{z}(h,\wh{\kappa}(h,\eta_n)) = \left(\frac{A}{\eta_n}\right)^{\eta_n} \prodlog{\frac{A}{\eta_n} e^{\frac{\check{B}(h)}{\eta_n}+\wt{\kappa}_n}}^{-\eta_n},
\end{equation*}
where $\check{B}(h) = -\frac{1}{2}V(h)'(P_I-P_U)V(h)$. Therefore,
\begin{equation*}
     \left(\frac{A}{\eta_n}\right)^{\eta_n} \prodlog{\frac{A}{\eta_n} e^{\frac{\check{B}(h)}{\eta_n}+\wt{\kappa}_n}}^{-\eta_n} \geq C(h)^{\eta_n} \quad \Leftrightarrow \quad \eta_n \prodlog{\frac{A}{\eta_n} e^{\frac{\check{B}(h)}{\eta_n}+\wt{\kappa}_n}} \leq \frac{A}{C(h)}.
\end{equation*}
But, as $n\to\infty$ this violates \eqref{E:larget_temper_4}, showing that $\wt{\kappa}_n\not\to\infty$ and finishing the proof.
\end{proof}

\begin{proof}[Proof of Proposition \ref{P:small_etaU}]

Throughout, $h$ is fixed. Also, we may assume $\eta_U < 1$, and hence $\wh{\kappa}(h) = \wh{\kappa}(h,\eta_U)$ from Theorem \ref{T:wU0_exists_gen} is the unique $\kappa$ enforcing \eqref{E:U_budget_new}.  Let us first assume there is a sequence $\eta_n \to 0$ and a number $\wt{\kappa}$ such that
\begin{equation*}
    \lim_{n\to\infty} \wt{\kappa}_n \dfn \wh{\kappa}(h,\eta_n) = \wt{\kappa}.
\end{equation*}
Given this, we first claim
\begin{equation}\label{E:small_eta_zlim}
    \begin{split}
    \lim_{n\to\infty} z(x,h,\wt{\kappa}_n) &= e^{\frac{1}{\alpha_I +\alpha_N}(\vartheta(x,h)+\wt{\kappa})^{-}};\quad \lim_{n\to\infty} z(x,h,\wt{\kappa}_n)^{1-\frac{1}{\eta_{n}}} = \frac{(\vartheta(x,h)+\wt{\kappa})^{+}}{\omega_U}.
    \end{split}
\end{equation}
Indeed, \eqref{E:z_crra} gives
\begin{equation*}
    z(x,h,\wt{\kappa}_n) = \left(\frac{A}{\eta_n}\right)^{\eta_n}\prodlog{\frac{A}{\eta_n} e^{\frac{B_n(x,h)}{\eta_n}}}^{-\eta_n},
\end{equation*}
where
\begin{equation*}
    A = \frac{\omega_U}{\alpha_I+\alpha_N};\qquad  B_n(x,h) = \frac{\vartheta(x,h)+\wt{\kappa}_n}{\alpha_I + \alpha_N}.
\end{equation*}
Next, for any $k$,  $\prodlog{z}^{-k} = z^{-k}e^{k\prodlog{z}}$ and hence
\begin{equation*}
    z(x,h,\wt{\kappa}_n) = e^{-B_n(x,h) + \eta_n \prodlog{\frac{A}{\eta_n} e^{\frac{B_n(x,h)}{\eta_n}}}}.
\end{equation*}
By assumption $B_n(x,h) \to B(x,h) = (\vartheta(x,h) + \wt{\kappa})/(\alpha_I + \alpha_N)$.  As l'Hopital's rule shows
\begin{equation*}
\lim_{y \to \infty} \frac{\prodlog{A y e^{y B}}}{y} = B^+;\qquad A>0,B\in\reals,
\end{equation*}
we  verify \eqref{E:small_eta_zlim} for $z$.  Next, the above identities also imply
\begin{equation}\label{E:small_eta_z_ident}
    z(x,h,\wt{\kappa}_n)^{1-\frac{1}{\eta_n}} = e^{-B_n(x,h) + \eta_n \prodlog{\frac{A}{\eta_n} e^{\frac{B_n(x,h)}{\eta_n}}}} \times \frac{\eta_n}{A} \prodlog{\frac{A}{\eta_n} e^{\frac{B_n(x,h)}{\eta_n}}},
\end{equation}
which in turn yields \eqref{E:small_eta_zlim} as
\begin{equation*}
    \lim_{n\to\infty}  z(x,h,\wt{\kappa}_n)^{1-\frac{1}{\eta_n}}  =  e^{\frac{1}{\alpha_I +\alpha_N}(\vartheta(x,h)+\wt{\kappa})^{-}} \frac{(\vartheta(x,h)+\wt{\kappa})^{+}}{\omega_U} = \frac{(\vartheta(x,h)+\wt{\kappa})^{+}}{\omega_U}.
\end{equation*}
Having obtained the limiting values in \eqref{E:small_eta_zlim} we now identify $\wt{\kappa}$.  To this end, re-writing \eqref{E:U_budget_new} and  using \eqref{E:ell_T_def} we have for each $n$
\begin{equation}\label{E:kappa_map}
    1 = \frac{\expvs{(\pi_0^U)'\Psi(X_T) \ell(T,X_T,h) z(X_T,h,\wt{\kappa}_n)}}{\expvs{\ell(T,X_T,h) z(X_T,h,\wt{\kappa}_n)^{1-\frac{1}{\eta_n}}}}.
\end{equation}
Using the explicit formula for $\ell(T,X_T,h)$ in \eqref{E:ell_T_def}, the upper bound for $z$ in Lemma \ref{L:clear_function_soln} and Assumption \ref{A:Psi}, it is clear that
\begin{equation*}
    \begin{split}
    &\lim_{n\to\infty} \expvs{(\pi_0^U)'\Psi(X_T) \ell(T,X_T,h) z(X_T,h,\wt{\kappa}_n)}\\
    &\qquad \qquad = \expvs{(\pi_0^U)'\Psi(X_T) \ell(T,X_T,h) e^{\frac{1}{\alpha_I +\alpha_N}(\vartheta(X_T,h)+\wt{\kappa})^{-}}} < \infty,
    \end{split}
\end{equation*}
where the inequality follows because, for $h$ fixed, $\vartheta(x,h)$ is bounded from below in $x$ (see \eqref{E:vartheta_to_V_ordering}).  As for the denominator we first note from \eqref{E:small_eta_z_ident} that
\begin{equation*}
    z(x,h,\wt{\kappa}_n)^{1-\frac{1}{\eta_n}} = z(x,h,\wt{\kappa}_n) \times \frac{\eta_n}{A} \prodlog{\frac{A}{\eta_n} e^{\frac{B_n(x,h)}{\eta_n}}}.
\end{equation*}
Using that $\prodlog{\cdot}$ is increasing, $\prodlog{z} \leq \log(z)$ for $z > e$ and $\prodlog{e} = 1$, we obtain
\begin{equation*}
    \frac{\eta_n}{A} \prodlog{\frac{A}{\eta_n} e^{\frac{B_n(x,h)}{\eta_n}}} \leq \frac{\eta_n}{A} 1_{\cbra{\frac{A}{\eta_n} e^{\frac{B^+_n}{\eta_n}} \leq e}} + \frac{\eta_n}{A}\left(\frac{B^+_n(x,h)}{\eta_n} + \log\left(\frac{A}{\eta_n}\right)\right)1_{\cbra{\frac{A}{\eta_n} e^{\frac{B^+_n}{\eta_n}} > e}}.
\end{equation*}
so that for $\eta_n < 1$
\begin{equation*}
    \frac{\eta_n}{A} \prodlog{\frac{A}{\eta_n} e^{\frac{B_n(x,h)}{\eta_n}}} \leq \frac{1}{A} + \frac{B^+_n(x,h)}{A}  + \left|\frac{\eta_n}{A}\log\left(\frac{\eta_n}{A}\right)\right|.
\end{equation*}
Assumption \ref{A:Psi}, the definition of $\ell(T,X_T,h)$ in \eqref{E:ell_T_def} and the upper bound for $z$ in Lemma \ref{L:clear_function_soln} allow us to invoke the dominated convergence theorem to obtain
\begin{equation*}
    \lim_{n\to\infty} \expvs{\ell(X_T) z(X_T,\wt{\kappa}_n)^{1-\frac{1}{\eta_n}}} = \frac{1}{\omega_U}\expvs{\ell(X_T)(\vartheta(X_T,h) + \wt{\kappa})^{+}}<\infty.
\end{equation*}
But, this implies
\begin{equation*}
    \frac{1}{\omega_U} = \frac{\expvs{(\pi_0^U)'\Psi(X_T) \ell(X_T) e^{\frac{1}{\alpha_I +\alpha_N}(\vartheta(x,h)+\wt{\kappa})^{-}}} }{\expvs{\ell(X_T)(\vartheta(X_T,h) + \wt{\kappa})^{+}}},
\end{equation*}
which is exactly \eqref{E:small_etaU_kappa_eqn} at $\wt{\kappa}$. However, as the right side in \eqref{E:small_etaU_kappa_eqn} is strictly decreasing in the argument $\kappa$, we see that $\wt{\kappa} = \wh{\kappa}(h,0)$ . Therefore, we have shown that for any $\eta_n \to 0$ such that $\wt{\kappa}_n = \wh{\kappa}(h,\eta_n)$ converges to some $\wt{\kappa}$, it must be that  $\wt{\kappa} = \wh{\kappa}(h,0)$.    Lastly, note that
\begin{equation*}
    S(t,x,h;\eta_n) = \frac{\condexpvs{\Psi(X_T) \ell(T,X_T,h) z(X_T,h,\wt{\kappa}_n)}{X_t = x}}{\condexpvs{\ell(T,X_T,h)z(X_T,h,\wt{\kappa}_n)}{X_t=x}}.
\end{equation*}
Using the exact same arguments as above for $(t,x,h)$ fixed, we may pass the limit through the expectation operator to obtain \eqref{E:small_eta_px_limit}.

The final item to prove is that $\eta_n \to 0$ implies $\wt{\kappa}_n \to \wh{\kappa}(h,0)$.  Now, we have already shown that every sub-sequence which converges to a finite limit, has the same limit point $\wh{\kappa}(h,0)$. It thus suffices show $\wt{\kappa}_n \to \pm \infty$ cannot happen.  To do this, we recall from the proof of Theorem \ref{T:wU0_exists_gen}, that for $\eta_U < 1$ the right hand side of \eqref{E:kappa_map} is strictly decreasing in $\kappa$.  Given this, assume that $\wt{\kappa}_n \to \infty$, and for a given $\kappa > 0$ assume $n$ is large enough so that $\wt{\kappa}_n \geq \kappa$. By monotonicity
\begin{equation*}
     1 \leq  \frac{\expvs{(\pi_0^U)'\Psi(X_T) \ell(T,X_T,h) z(X_T,h,\kappa)}}{\expvs{\ell(T,X_T,h) z(X_T,h,\kappa)^{1-\frac{1}{\eta_n}}}}.
\end{equation*}
With $\kappa$ fixed, we repeat the above analysis taking $\eta_n \to 0$ to obtain
\begin{equation*}
     1 \leq  \frac{\omega_U \expvs{(\pi_0^U)'\Psi(X_T) \ell(X_T) e^{\frac{1}{\alpha_I +\alpha_N}(\vartheta(x,h)+\kappa)^{-}}} }{\expvs{\ell(X_T)(\vartheta(X_T,h) + \kappa)^{+}}}.
\end{equation*}
As this holds for any $\kappa > 0$ we take $\kappa\to\infty$ and use monotone convergence to obtain $1\leq 0$, a contradiction. Therefore, $\wh{\kappa}_n \not\to \infty$.   Similarly, if $\wt{\kappa}_n \to -\infty$, fix $\kappa > 0$ and take $n$ large enough so that $\wt{\kappa}_n \leq -\kappa$. By monotonicity
\begin{equation*}
     1 \geq  \frac{\expvs{(\pi_0^U)'\Psi(X_T) \ell(T,X_T,h) z(X_T,h,-\kappa)}}{\expvs{\ell(T,X_T,h) z(X_T,h,-\kappa)^{1-\frac{1}{\eta_n}}}}.
\end{equation*}
Taking $\eta_n \to 0$ gives
\begin{equation*}
     1 \geq  \frac{\omega_U \expvs{(\pi_0^U)'\Psi(X_T) \ell(X_T) e^{\frac{1}{\alpha_I +\alpha_N}(\vartheta(x,h)-\kappa)^{-}}} }{\expvs{\ell(X_T)(\vartheta(X_T,h) -\kappa)^{+}}}.
\end{equation*}
Taking $\kappa\to\infty$ and use monotone convergence to obtain $1\geq \infty$, a contradiction. Therefore, $\wh{\kappa}_n \not\to -\infty$ and the proof is complete.
\end{proof}

\bibliographystyle{siam}

\bibliography{master}

\end{document}